\documentclass[pra]{revtex4}

\usepackage{amsmath, amssymb, amsthm, amsfonts, latexsym, graphicx}
\usepackage{float}
\usepackage{xcolor}
\usepackage{subfig}
\usepackage{verbatim} 
\usepackage[colorlinks]{hyperref}
\usepackage{hypcap}
\newcommand{\ket}[1]{\left| #1\right\rangle}        
\newcommand{\bra}[1]{\left\langle #1\right|}        
\newcommand{\braket}[2]{\left\langle #1 | #2 \right\rangle} 
\newcommand{\ketbra}[2]{\ket{#1}\!\bra{#2}}
\newcommand{\mvalue}[3]{\left\langle #1 \left| #2 \right| #3 \right\rangle}
\newcommand{\ii}{\mathbb{I}}     
\newcommand{\norm}[1]{\left\| #1\right\|}        
\newtheorem{definition}{Definition}

\newcommand{\alphavar}{\mu}
\newcommand{\Lambdavar}{\Lambda}
\newcommand{\sparse}[1]{$#1$--sparse}
\newcommand{\initt}{t_0}
\newcommand{\dt}{\Delta t}
\newcommand{\numBitsInT}{N_T}
\newcommand{\numQubitsInH}{n_H}
\newcommand{\Noracle}{N_{\rm queries}}
\newcommand{\Costvar}{C}

\newtheorem{theorem}{Theorem}

\newtheorem{corollary}[theorem]{Corollary}

\newcommand{\eq}[1]{\hyperref[eq:#1]{(\ref*{eq:#1})}}
\renewcommand{\sec}[1]{\hyperref[sec:#1]{Section~\ref*{sec:#1}}}
\newcommand{\app}[1]{\hyperref[app:#1]{Appendix~\ref*{app:#1}}}
\newcommand{\tab}[1]{\hyperref[tab:#1]{Table~\ref*{tab:#1}}}
\newcommand{\fig}[1]{\hyperref[fig:#1]{Figure~\ref*{fig:#1}}}
\newcommand{\thm}[1]{\hyperref[thm:#1]{Theorem~\ref*{thm:#1}}}
\newcommand{\lem}[1]{\hyperref[lem:#1]{Lemma~\ref*{lem:#1}}}
\newcommand{\cor}[1]{\hyperref[cor:#1]{Corollary~\ref*{cor:#1}}}
\newcommand{\defn}[1]{\hyperref[def:#1]{Definition~\ref*{def:#1}}}
\newcommand{\alg}[1]{\hyperref[alg:#1]{Algorithm~\ref*{alg:#1}}}
\newcommand{\defin}[1]{\hyperref[def:#1]{Definition~\ref*{def:#1}}}

\begin{document}

\title{On The Power Of Coherently Controlled Quantum Adiabatic Evolutions}
\author{M\'{a}ria Kieferov\'{a}}
\affiliation{Research Center for Quantum Information, Institute of Physics,
Slovak Academy of Sciences, Bratislava,
Slovakia}
\affiliation{Department of Theoretical Physics, Comenius University, Bratislava,
Slovakia}
\author{Nathan Wiebe}
\affiliation{Quantum Architectures and Computation Group, Microsoft Research,
Redmond, WA 98052, USA}
\affiliation{Institute for Quantum Computing and Department of Combinatorics and Optimization, University of Waterloo,
200 University Ave., West, Waterloo, Ontario, Canada}
\begin{abstract}
A major challenge facing adiabatic quantum computing is that
algorithm design and error correction can be difficult for adiabatic
quantum computing. Recent work has considered addressing this challenge
by using coherently controlled adiabatic evolutions in the place of classically
controlled evolution.  An important question remains: what is the relative power
of controlled adiabatic evolution to traditional adiabatic evolutions?
We address this by showing that coherent control and measurement provides a way to
average different adiabatic evolutions in
ways that cause their diabatic errors to cancel, allowing for adiabatic
evolutions to combine the best characteristics
of existing adiabatic optimizations strategies that are mutually
exclusive in conventional adiabatic QIP.  This result shows that coherent control and
measurement can provide advantages for adiabatic state preparation.  
We also provide upper bounds on the complexity of simulating
such evolutions on a circuit based quantum computer and provide sufficiency
conditions for the equivalence of  controlled adiabatic evolutions to adiabatic
quantum computing.
\end{abstract}
\date{\today}
\maketitle

The quantum adiabatic theorem is an essential tool for quantum information
processing and quantum control~\cite{oreg1984adiabatic,kuklinski1989adiabatic,FG00,biamonte2011adiabatic,FW13,MG14}. 
It states that the evolution generated by a slowly varying Hamiltonian (relative
to the minimum eigenvalue gap) maps eigenstates
of the
initial Hamiltonian to eigenstates of the final Hamiltonian~\cite{kato50}.  This process
provides a simple and error robust method for state preparation that
is used extensively in quantum simulation, adiabatic quantum computing as well
as pulse design.  A drawback of adiabatic
evolution is that it is often much slower than competing state preparation
methods. Finding ways of reducing
``diabatic'' errors (which result from using a finite evolution time) is vitally important for practical applications
of adiabatic state preparation.

Two major strategies have been proposed to minimize the error in adiabatic
evolutions: local adiabatic evolution and
boundary cancellation methods.
Local adiabatic evolution~\cite{RC02,VDM+01} (LAE) minimizes the time to reach the
adiabatic regime by choosing the evolution speed such that the adiabatic condition
is satisfied at each instant throughout the evolution. In a typical scenario of LAE, the rate at which the
Hamiltonian changes is fast in the
beginning and the end of evolution, when the distance between the ground state
and the first excited state is large, and
small in the middle around the minimal gap. This approach optimizes the scaling
of the evolution time with the size of
the system and works best to reduce diabatic errors in the short time or
``Landau--Zener'' regime (so called because the
Landau--Zener formula provides a better approximation to the resultant state
than adiabatic perturbation theory does).

Boundary cancellation methods minimize the error
in the adiabatic approximation once the adiabatic
condition is met~\cite{LRH09,WB12}. These methods polynomially improve the
error scaling, relative to LAE, by setting the first
$n-1$ derivatives of the Hamiltonian to zero at the boundaries. This strategy tends to lead to
taking the Hamiltonian to be slowly varying near the beginning and end of the evolution, which
typically is where the eigenvalue gap is largest.  Since the
Hamiltonian will often vary slowly when the gap is large, it forces the evolution to speed up around
the minimal gap, which retards the
convergence to the adiabatic regime (the regime where adiabatic perturbation theory
applies).

These two approaches are typically at odds: LAE says that you should move
quickly when the gap is large to
minimize the error, which is often at the beginning and end of the adiabatic
passage~\cite{VDM+01,RC02}; whereas boundary cancellation
methods show that it is often best to move very slowly at the beginning and end
of the evolution~\cite{LRH09,RPL10}.  The question is: can
these two objectives be simultaneously satisfied and if so how?

We consider a model of adiabatic quantum computation that can achieve both
goals.  Our hybrid  model for adiabatic
computation uses a small control register that the user has universal control over,
and a larger adiabatic system that is coherently controlled by the smaller register. These generalizations
grant us two freedoms: (a) the adiabatic subsystem can evolve under a
superposition of different adiabatic evolutions
and (b) measurement can be used on the control qubits without exciting the
system out of the groundstate. 
These freedoms allow us
to escape the
constraints of unitarity and implement a wider class of
operations including linear combinations of
unitaries~\cite{CW12}, which we use to increase the resilience of the evolution
to diabatic errors. This model also subsumes those of~\cite{Hen14,ZR99}.

Unlike the previous methods, we do not search for a single optimal adiabatic
evolution. Instead we take two (or more)
evolutions that generate errors that are oriented in opposite directions as
in~\fig{diagram} and then use the
non--deterministic circuit in \fig{circuit} to suppress these errors by
performing an appropriate weighted average
of the evolutions. 
We then show that a linear combination of adiabatic evolutions can
asymptotically decrease the error in the adiabatic
approximation. The resultant averaged adiabatic evolution can have the benefits
of both LAE and boundary cancellation methods:
the convergence to the
adiabatic regime is comparable to LAE, while the error scaling in the adiabatic
regime is comparable to that of boundary
cancellation methods.

In the following section we review the adiabatic theorem.
We then provide the gadget that we use to cancel the leading order
diabatic errors in~\sec{gadget}.  We illustrate the utility of this method in~\sec{single}
where we apply the gadget to approximately cancel the dominant diabatic transition.
We provide methods in~\sec{general} that 
simultaneously suppress every transition, assuming that the adiabatic paths obey a
particular symmetry condition.  Finally, we discuss how our techniques can combine
the best features of local--adiabatic evolution and boundary cancellation methods in~\sec{LA} and
then discuss the implementation of our model of coherently controlled adiabatic
evolution using a quantum computer in~\sec{qcimp}.

\section{Review of the adiabatic theorem}
It is not possible to provide a closed form solution to the Schr\" odinger equation for the case of time--dependent Hamiltonians in general.
It is customary in such cases to express the time evolution operator, which is the formal solution to
\begin{equation}
\frac{\partial U(t,0)}{\partial t}  = -i H(t) U(t,0),
\end{equation}
as 
\begin{equation}
U(t,0) = \mathcal{T} e^{-i \int_0^t H(\tau) \mathrm{d} \tau}:= \lim_{r\rightarrow \infty} \prod_{j=0}^{r-1} e^{-i H(jt/r)t/r}.
\end{equation}
A wide array of approximation methods exist for $U(t)$ including the Magnus expansion~\cite{Mag54}, Dyson series~\cite{Joa75}, Floquet
theory~\cite{LBM+95}, the Landau--Zener formula~\cite{Joy94} and the adiabatic approximation.

The adiabatic approximation is widely used to approximate quantum dynamics in cases where rate of change of the time--dependent Hamiltonian is slow
relative to an appropriate power of the minimum eigenvalue gap.  In essence, the approximation states that if you prepare a system in an eigenstate of
the Hamiltonian and evolve sufficiently slowly then the quantum system will remain in the eigenstate throughout the evolution.  This lack of
excitation throughout the process makes it analogous to reversible adiabatic processes in statistical mechanics.  This analogy is not exact since the
change in Von--Neumann entropy is also zero for any unitary process and so the ``adiabatic" moniker persists for largely historical reasons.

Since the adiabatic approximation requires slow evolution, it is useful to consider how the approximation error scales as the speed of the transition
from the initial to the final Hamiltonian decreases.  This makes it natural to parameterize time via the variable $s$ where
\begin{equation}
s=t/T,
\end{equation}
and $T$ is the total time for the adiabatic passage. While an adiabatic evolution occurs on $t\in [0,T]$, $s\in [0,1]$ regardless
of the actual duration of the evolution.  This means that if the Hamiltonian is re--parameterized as $H(s)$, then we can increase $T$ to make the
evolution slower without fundamentally changing the form of the evolution.

We need to introduce some further notation before we can discuss the adiabatic approximation in greater detail.  Firstly we define $\ket{n(s)}$ to be
the instantaneous eigenvectors of the time--dependent Hamiltonian:
\begin{equation}
H(s)\ket{n(s)}=E_n(s) \ket{n(s)},
\end{equation}
and we make no assumptions about the ordering of $E_n$ (i.e. we do not assume that $E_0 \le E_1$).
Also for notational simplicity, we define $\ket{g(s)}:=\ket{E_0(s)}$.  We refer to this state as $\ket{g(s)}$ because it will represent the ground
state in many practical examples of adiabatic QIP.  The eigenvalue gaps will also be key to our analysis and so we use the following notation for
them:
\begin{equation}
\gamma_{\mu,\nu}(s) := E_{\mu}(s)-E_{\nu}(s).
\end{equation}

The adiabatic approximation is often expressed in many different ways.  The simplest of these states that
\begin{equation}
U(1,0)\ket{g(0)}\approx e^{-i\int_0^1  E_0\left( \xi\right) T d
\xi}\ket{g(1)}+O\left(\frac{1}{T}\right)
\label{zero_ord}
\end{equation}
In general the adiabatic approximation holds if
\begin{equation}
T \gg \frac{\max_{s} (\|\partial_s H(s)\|^a+\|\partial_s^2 H(s)\|^b+\|\partial_s^3 H(s)\|^c)}{\min_s\gamma_{g,1}(s)^d},\label{eq:adiabaticcond}
\end{equation}
for integers $a,b,c$ and $d$ that depend on the properties of the Hamiltonian~\cite{adiab_error,JSR,teufel2003adiabatic,SL05}.
A common misconception is that the adiabatic approximation holds if
\begin{equation}
T \gg  \frac{\max_s\|\partial_s H(s)\|}{\min_s\gamma_{0,1}(s)^2},
\end{equation}
although this criteria is appropriate for sufficiently slow evolutions under smoothly varying Hamiltonians~\cite{adiab_error,JSR}.

We refer to such results as zeroth order adiabatic theorems, because they provide an estimate of the error that is correct to zeroth order in
$T^{-1}$, meaning that
they simply tell you that the error is zero if the adiabatic process is infinitely slow.  In order to show that we can combine different adiabatic
evolutions together to cancel the error, we need to have a sharper adiabatic condition that approximates the error to at least $O(1/T)$.
It is necessary for us to use a first order adiabatic approximation, which provides us with the error in the adiabatic approximation correct to
$O(1/T^2)$~\cite{adiab_error}:
\begin{equation}
\left.
\left(\ii-\ket{g\left(1\right)}\!\bra{g\left(1\right)}\right)U\left(1,
0\right)\ket{g\left( s \right) }=
\sum_{n\neq g}e^{-i\int_0^1 E_n\left(\xi\right) T d
\xi}\frac{\braket{\dot{n}(s)}{g(s)}e^{i\int_0^s\gamma_{g,n}(\xi)d\xi
T}}{-i\gamma_{g,n}\left(s\right)T}\right|_{s=0}^1
\ket{n\left(1\right)} +O\left(\frac{1}{T^2}\right) \label{eq:error}.
\end{equation}
This result can easily be found by using path integral methods using techniques also discussed in~\cite{MMP06,ghosh2013high,farhi1992functional} and
upper bounds on the magnitude of the sum of all $O(1/T^2)$ terms are given in~\cite{adiab_error}.

Equation~\eq{error} tells us something surprising: the leading order contribution to the error in the adiabatic approximation does not depend on the
minimum gap but rather the eigenvalue gap at the beginning and the end of the evolution, which motivates taking $\dot H(s)=0$ or equivalently
$\braket{\dot{n}(s)}{g(s)}=0$ on the boundary as per boundary cancellation methods~\cite{LRH09}.  The apparent contradiction posed by~\eq{error} is easily
resolved.  Adiabatic
conditions like~\eq{adiabaticcond} give criteria for the convergence of the adiabatic perturbation series of $U(T,0)$ in powers of $1/T$ and equations
such as~\eq{error} give a truncated expression for the power series.  This means that after a critical evolution speed, the error in the adiabatic
approximation no longer depends on the minimum gap; whereas, the error depends crucially on the minimum gap before this point.  We refer to the regime
where
the minimum gap dictates the error as the Landau--Zener regime and the regime where it does not as the adiabatic regime.

Similarly, the first order adiabatic theorem relies on several conditions outlined in~\cite{adiab_error}. First, the Hamiltonian must be
twice differentiable and 3-times piecewise
differentiable with all such derivatives upper bounded by a constant. Second, the system must be already in the adiabatic regime (i.e. the
$\Theta(1/T)$ contribution to the error is much greater than the sum of all $O(1/T^2)$ contributions). Third, we require that the norm of the
Hamiltonian
be upper-bounded by a constant for all times during the evolution.  These criteria guarantee the validity of \eq{error}.

A common way to reduce errors in both the Landau--Zener regime, as well as the adiabatic regime, is to change the path used in the adiabatic
evolution.
The most frequently used adiabatic path, known as linear interpolation, is
\begin{equation}
H(s)= (1-s) H_0 + sH_1,\label{eq:linear}
\end{equation}
where $H_0$ is the initial Hamiltonian and $H_1$ is the final Hamiltonian.  There are of course many ways that one could imagine transitioning from
the initial Hamiltonian to the final Hamiltonian.  Each of these ways represents a particular ``adiabatic path'' and~\eq{linear} is known as the
linear adiabatic path.  More generally we could consider a path of the form
\begin{equation}
H(s)= (1-f(s)) H_0 + g(s)H_1,\label{eq:nonlinear}
\end{equation}
where $f(0)=g(0)=0$ and $f(1)=g(1)=1$.  Such paths can be extremely important for adiabatic quantum computing because they allow the evolution
to slow down through, or even avoid,  parts of the evolution that contribute substantially to the error; however, here we assume the simple case of
$g(s)=f(s)$.  We do
 not require that the range of $f$ is $[0,1]$ here.  In fact, some of the adiabatic paths that we consider will attain negative values
and values greater than $1$.

Other examples of non--linear paths include local--adiabatic evolution, which seeks to minimize the error in the Landau--Zener regime by choosing the
evolution speed to be smallest near the minimum gap.  Boundary cancellation methods on the other hand choose paths that minimize the error in the
adiabatic regime by choosing the evolution speed to be zero at $s=0$ and $s=1$.  These two strategies are seemingly orthogonal.  At present there is
no known method that combines the best features of local adiabatic paths and the paths yielded by boundary cancellation methods.  Our work provides a
way  to achieve this, thereby illustrating that controlled adiabatic evolution affords greater power than conventional adiabatic evolution.

\section{Controlled adiabatic evolution using a small number of ancillas\label{sec:gadget}}
The central idea behind our approach is to use a gadget that was recently proposed in~\cite{lin_comb} to non--deterministically implement
the weighted average of two or more adiabatic evolutions.  
This idea of using controlled adiabatic evolutions and measurement has been recently explored by Itay Hen~\cite{Hen14} and is also used in holonomic
quantum computing~\cite{ZR99}; however, these results do not consider using coherent control and measurement to suppress diabatic errors.
The gadget that we use for this averaging process is given in~\fig{circuit}.
The circuit in~\fig{circuit} probabilistically implements linear combinations of unitary operations, as seen through the following argument
\begin{align}
\ket{\psi}\ket{0} &\rightarrow \ket{\psi}\left(
\cos{\theta}\ket{0}+\sin{\theta}\ket{1} \right)\nonumber\\
&\rightarrow \cos{\theta}U_A
\ket{\psi}\ket{0}+\sin{\theta}U_B\ket{\psi}\ket{1}\nonumber\\
&\rightarrow \left(\cos^2{\theta}U_A+\sin^2{\theta}U_B\right)
\ket{\psi}\ket{0}+\sin{\theta}\cos{\theta}\left(U_B-U_A\right)\ket{\psi}\ket{1}
\label{lin_comb}.
\end{align}
We then see that if the ancilla register is measured to be $0$ then the circuit performs a weighted combination of $U_A$ and $U_B$ on the state
$\ket{\psi}$ otherwise the circuit implements the difference between the two operators.
\begin{equation}
p(0)\ge 1 - \norm{ \left( U_A-U_B\right) \ket{\psi} }^2.\label{eq:pfail}
\end{equation}
The generalization to cases where multiple $U_A$ and $U_B$ are used is trivial, it simply involves increasing the number of qubits used to control the
overall rotation~\cite{lin_comb}.
Such circuits, or variants thereof, are also used in \cite{floating_point,Paetznick2013}.

For the case of adiabatic evolution, we know that to zeroth order
\begin{align}
U_A(T_A,0)\ket{g(0)} = e^{-i\int_0^1 E_0(f_A(s)) \mathrm{d}s T_A}\ket{g(1)} + O(1/T),\nonumber\\
U_B(T_B,0)\ket{g(0)} = e^{-i\int_0^1 E_0(f_B(s)) \mathrm{d}s T_B}\ket{g(1)} + O(1/T),\label{eq:nophases}
\end{align}
where $T = \max\{T_A,T_B\}$.  This means that, to leading order, both $U_A$ and $U_B$ generate the same evolution up to a global phase and hence we
expect the success probability to be high if the
phases picked up by $\ket{g}$ under both evolutions are comparable.

Rather than choosing different paths that apply the same phase to $\ket{g(s)}$, we counter--rotate the evolution of each eigenstate by including an
additional phase to each unitary.  This affords us much greater freedom to choose adiabatic paths for $U_B$ and $U_A$.  In particular, we choose these
phases such that
\begin{align}
U_A(T_A,0)\ket{g(0)} = e^{i\int_0^1 E_0(f_A(s)) \mathrm{d}s T_A}\left(\mathcal{T}e^{-i\int_0^1 H(f_A(s)) \mathrm{d}s T_A}\ket{g(1)}\right),\nonumber\\
U_B(T_A,0)\ket{g(0)} = e^{i\int_0^1 E_0(f_B(s)) \mathrm{d}s T_B}\left(\mathcal{T}e^{-i\int_0^1 H(f_B(s)) \mathrm{d}s
T_A}\ket{g(1)}\right).\label{eq:phases}
\end{align}
We see from the choices of phases in~\eq{phases} that~\eq{pfail} gives the failure probability of the linear combination $O(1/T)$ in the
limit of large $T$.  This means that
the failure probability will typically be extremely small for adiabatic processes and even if a failure is observed, the gadget in~\fig{circuit}
informs the user that a failure has occurred and the
state preparation process can be repeated until success is obtained.  We see from numerical experiments that the failure probability of these circuits
has a near--negligible impact on the cost of coherently controlled adiabatic state preparation in the adiabatic regime.

\begin{figure}
\centering
\begin{minipage}{0.45\linewidth}
\centering
\vspace{0.65 in}
\includegraphics[width=15em]{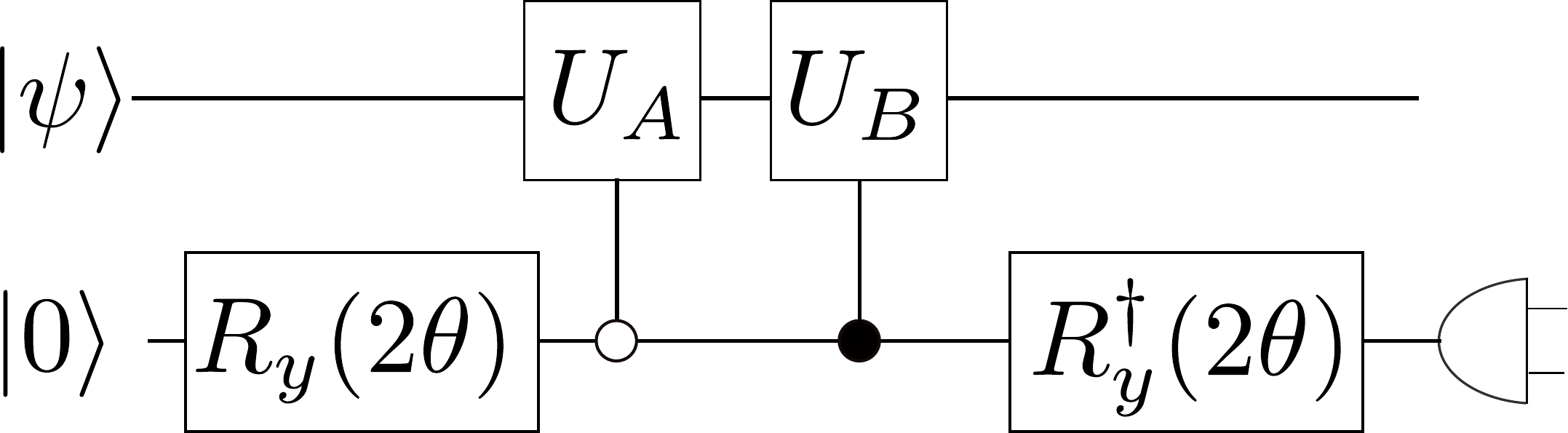}
\vspace{0.3 in}
\caption{Circuit for linear combination of two unitary operations.}
\label{fig:circuit}
\end{minipage}
\hspace{0.5 cm}
\begin{minipage}{0.45\linewidth}
\includegraphics[width=0.5\linewidth]{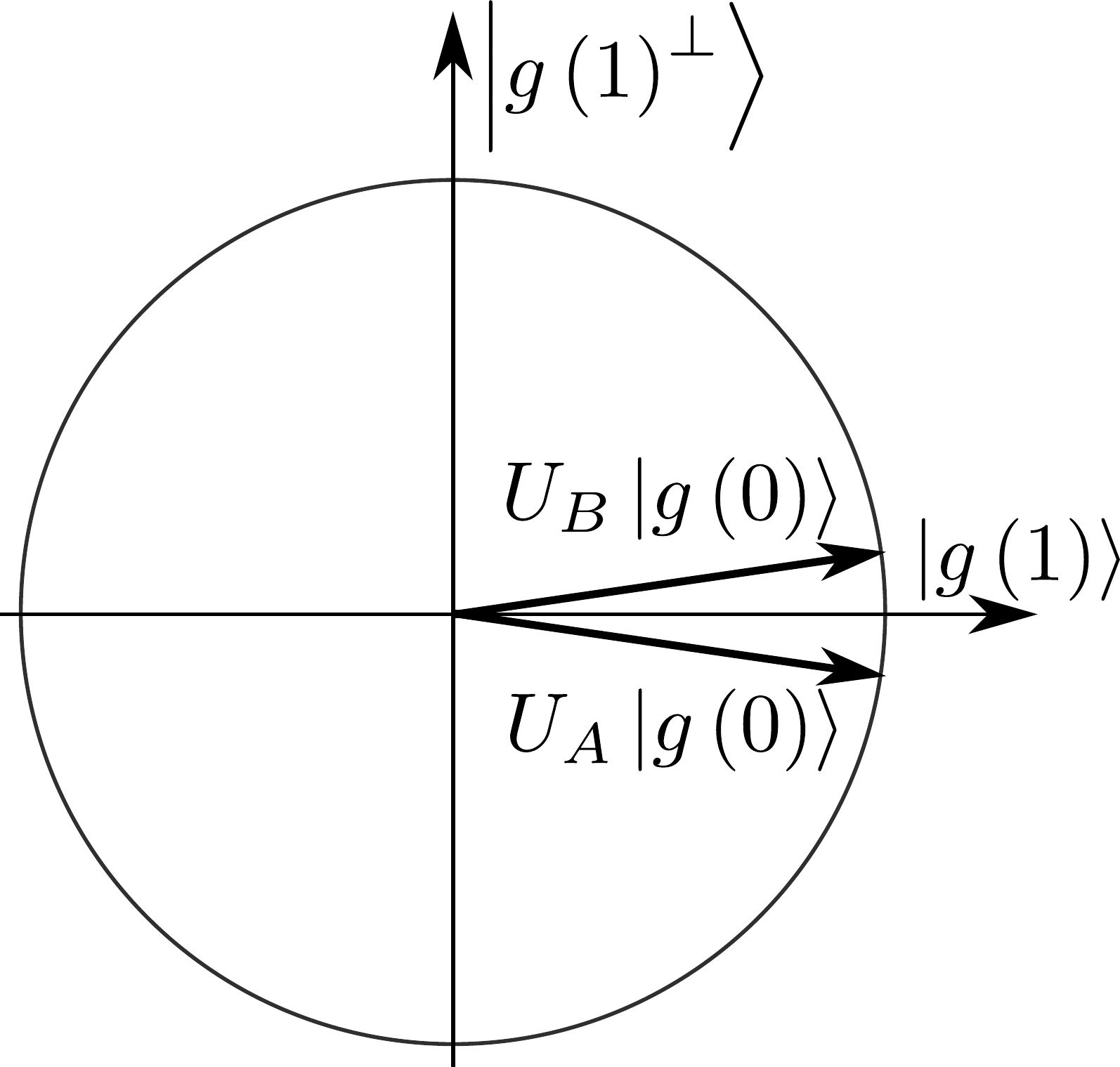}
\caption{The average of two evolutions with opposite errors will completely
eliminate the first order error.}\label{fig:diagram}
\end{minipage}
\hspace{0.5cm}
\end{figure}

Generalization of these ideas to cases where more than two unitary evolutions are averaged is straight forward and is discussed in detail
in~\cite{CW12}.
We present the two--unitary case explicitly here since the majority of our results focus on averaging two different adiabatic evolutions.

\section{A general method for canceling a single transition\label{sec:single}}
Our first approach is a generalization of the strategy employed by Wiebe and Babcock in~\cite{nathans}, which suppresses
the dominant transition in the adiabatic passage for adiabatic paths satisfying
\begin{equation}
\left. \frac{\braket{\dot{n}(s)}{g(s)}}{\gamma_{g,n}\left(s\right)}\right|_{s=0} = \left.
\frac{\braket{\dot{n}(s)}{g(s)}}{\gamma_{g,n}\left(s\right)}\right|_{s=1},\label{eq:WBsymmetry}
\end{equation}
by choosing the evolution time $T$ appropriately.
Our strategy is to suppress a single transition, not by choosing a single time and requiring a symmetry condition as per~\cite{nathans}, but by
interfering the adiabatic evolution with a dual evolution as suggested in~\fig{diagram}.  This allows such errors to be suppressed
for any evolution time and any primary path.  We also provide a method for suppressing the two most significant diabatic transitions
in~\app{threelevel}.

We wish to choose, for fixed $H_A$, an adiabatic path that quadratically suppresses the transition $\ket{g(0)} \rightarrow \ket{e(1)}$ where
$\ket{e(s)}$ is any given instantaneous eigenstate of $H(s)$.
From~\eq{error} and~\eq{phases}, we see that if we combine $U_A(T_A,0)$ with $U_B(T_B,0)$ and achieve a successful measurement outcome then we obtain
a result proportional to
\begin{equation}
(\cos(\theta)^2U_A(T_A,0) +\sin(\theta)^2U_B(T_B,0))\ket{g(0)} = \ket{g(1)} + O(1/T):=\ket{\phi}.
\end{equation}
So to leading order, the linear combination will give the correct result.
Then using,~\eq{error} it is clear that
\begin{align}
\ketbra{e(1)}{e(1)}\ket{\phi}\propto&\cos^2{(\theta)}
\left[\frac{\braket{\dot{e}^A_1(1)}{g(1)}}{-i\gamma^A_{g,e}\left(1\right)T_A}
-e^{+i\int_0^1 \gamma^A_{g,e^A}\left(\xi\right)T_A d \xi}
\frac{\braket{\dot{e}^A(0)}{g(0)} 
}{-i\gamma^A_{g,e}\left(0\right)T_A}\right]\nonumber \\
&+\sin^2{(\theta)}
\left[\frac{\braket{\dot{e}^B(1)}{g(1)}}{-i\gamma^B_{g,e}\left(1\right)T_B} - 
e^{+i\int_0^1 \gamma^B_{g,e}\left(\xi\right)T_B d \xi}  
\frac{\braket{\dot{e}^B(0)}{g(0)}}{-i\gamma^B_{g,e}\left(0\right)T_B}\right]\nonumber\\
&+O(1/T^2).\label{eq:errorxpr}
\end{align}
This transition can therefore be canceled, to $O(1/T^2)$, by choosing $\theta, T_B$ and $f_B$ such that the weighted average of the diabatic
transitions to the state $\ket{e}$ is
zero:

\begin{align}
0=&\cos^2{(\theta)}
\left[\frac{\braket{\dot{e}^A_1(1)}{g(1)}}{-i\gamma^A_{g,e}\left(1\right)T_A}
-e^{+i\int_0^1 \gamma^A_{g,e^A}\left(\xi\right)T_A d \xi}
\frac{\braket{\dot{e}^A(0)}{g(0)} 
}{-i\gamma^A_{g,e}\left(0\right)T_A}\right]\nonumber \\
&+\sin^2{(\theta)}
\left[\frac{\braket{\dot{e}^B(1)}{g(1)}}{-i\gamma^B_{g,e}\left(1\right)T_B} - 
e^{+i\int_0^1 \gamma^B_{g,e}\left(\xi\right)T_B d \xi}  
\frac{\braket{\dot{e}^B(0)}{g(0)}}{-i\gamma^B_{g,e}\left(0\right)T_B}\right]
\label{eq:errors_2levels}
\end{align}
where
$\braket{\dot{e}(s)}{g(s)}=\frac{\mvalue{e(s)}{\dot{H}(s)}{g(s)}}{\gamma_{e,g}
(s)}$. 
Thus it is reasonable to expect that this condition can be met by choosing $\theta$ and $f_B$ properly.  The remaining question is: how can this be
done in practice?  We provide two strategies for finding $f_B$ for any fixed $f_A$ such that these errors cancel to leading order.

\subsection{Partially Anti--Symmetric Combination\label{sec:partial}}
Our first method chooses the paths $f_A$ and $f_B$ to satisfy an anti--symmetric condition on the derivatives at the beginning and end of the
evolution.  This 
approach is most useful in cases where it is desirable for $f_B$ to be as similar to $f_A$ as possible.  
When optimizing these paths it is important to note that although $f_A$ and $f_B$ are arbitrary interpolation functions that describe the adiabatic
paths, they are constrained to obey 
\begin{align}
&f_A(0)=f_B(0)=0,\\
& f_A(1)=f_B(1)=1.
\end{align}
 Furthermore, 
let us choose $f_B$ such that its derivatives are symmetric with $f_A$ at $s=0$ and anti--symmetric at $s=1$
\begin{align}
\left.\dot{f}_B\left(s\right)\right|_{s=0}&=\left.
\dot{f}_A\left(s\right)\right|_{s=0},\nonumber\\
\left.\dot{f}_B\left(s\right)\right|_{s=1}&=-\left.\dot{f}
_A\left(s\right)\right|_{s=1}.\label{eq:fbderivs}
\end{align}
 Then using~\eq{fbderivs},
\eq{errors_2levels} simplifies to
\begin{align}
\left( \frac{\cos^2{(\theta)}}{T_A} -
\frac{\sin^2{(\theta)}}{T_B}\right)\frac{\braket{\dot{e}(1)}{g(1)}}{\gamma_{g,e}
\left(1\right)} 
=\left( \frac{\cos^2{(\theta)}}{T_A}e^{+i\int_0^1
\gamma^A_{g,e}\left(\xi\right)T_A d \xi} + \frac{\sin^2{(\theta)}}{T_B}
e^{+i\int_0^1 \gamma^B_{g,e}\left(\xi\right)T_B d \xi} \right)
\frac{\braket{\dot{e}(0)}{g(0)}}{\gamma_{g,e}\left(0\right)}.\label{eq:23cond}
\end{align}
Equation~\eq{23cond} can be satisfied for any $f_A$ and $T_A$ by setting
\begin{align}
T_B&=\frac{\int_0^1 \gamma^A_{g,e}\left(\xi\right) d \xi T_A
+\left(2n+1\right)\pi}{\int_0^1 \gamma^B_{g,e}\left(\xi\right) d
\xi} \label{2level_time}\\
\theta&=\arctan\left({\sqrt{\frac{T_B}{T_A}}}\right).\label{eq:2level_weight}
\end{align}
This solution reduces to that of~\cite{nathans} in the limit as $T_A\rightarrow 0$; however, a non--trivial secondary path will always be needed if
the symmetry condition demanded by~\cite{nathans} is not held.

An important consequence of taking the derivatives to be negative at $s=1$ is that there exists $s'$ such that
$f(x)>1$ for all $x\in (s',1)$.  This is a consequence of the fact that $H(s)$ is twice differentiable and hence $f'_B$ is continuous; from which the
result directly
follows from the mean value theorem.  Thus $f_b$ does not monotonically approach $1$ as $s\rightarrow 1$, but rather it overshoots the value and then
reverses direction to end the evolution at $s=1$.  Such reversals of direction are analogous to the backwards time steps used in Trotter--Suzuki
methods and, although non--traditional, are not necessarily problematic for adiabatic evolution.

We see from this discussion that controlled adiabatic paths can be used to suppress diabatic errors in ways that are impossible using traditional
adiabatic optimization strategies.  In particular, for any optimization strategy, such as local adiabatic evolution, we can always find a second path
to add to the primary path to suppress a chosen transition to one order higher.  These ideas can also be generalized to suppress more than one
transition; however, finding a closed form solution is difficult in such cases.  We discuss generalizing this method to simultaneously suppress two
diabatic transitions in~\app{threelevel}.
A drawback of this approach is that it cannot be used
for arbitrary small times because
\eqref{2level_time} forces a
difference at least $\frac{\pi}{\int_0^1 \gamma_{g,e}\left(\xi\right)}$ between
evolution times.  We address this issue below by providing a method that does not require a shift in time, but requires a more substantial deformation
to the primary adiabatic path.

\subsection{Completely Anti--Symmetric Combination\label{sec:complete}}
An alternative approach is to set the derivatives for the second path to be completely
antisymmetric
\begin{align}
\left. \dot{f}_B\left(s\right) \right|_{s=0}=\left.
-\dot{f}_A\left(s\right)\right|_{s=0}\nonumber\\
\left.
\dot{f}_B\left(s\right)\right|_{s=1}=\left.-\dot{f}_A\left(s\right)\right|_{s=1}\label{eq:completeSymmetry}.
\end{align}

Pluging~\eq{completeSymmetry} into \eq{errors_2levels} we obtain

\begin{align}
\left( \frac{\cos^2{(\theta)}}{T_A} -
\frac{\sin^2{(\theta)}}{T_B}\right)\frac{\braket{\dot{e}(1)}{g(1)}}{\gamma^B_{g,
e }
\left(1\right)} 
=\left( \frac{\cos^2{(\theta)}}{T_A}e^{+i\int_0^1
\gamma^A_{g,e}\left(\xi\right)T_A d \xi} - \frac{\sin^2{(\theta)}}{T_B}
e^{+i\int_0^1 \gamma^B_{g,e}\left(\xi\right)T_B d \xi} \right)
\frac{\braket{\dot{e}(0)}{g(0)}}{\gamma^B_{g,e}\left(0\right)}\label{eq:2ears}.
\end{align}
The error is suppressed when \eq{2level_weight} holds and 
\begin{equation}
T_B=\frac{\int_0^1 \gamma^A_{g,e}\left(\xi\right)T_A d \xi
+ 2n\pi}{\int_0^1 \gamma^B_{g,e}\left(\xi\right) d \xi }.
\end{equation}
In other words, the gap integrals for both paths must be equivalent modulo $2\pi$. This
removes the difficulty with offsetting one of times.
However, in this case, the path $f_B$ both begins and ends the evolution by moving backwards. Alternatively, we can modify one boundary from each
path. This backwards motion at $s=0$ means that
there exists $\delta>0$ such that the range of $f_B(s)$ is within $[-\delta, 1+\delta]$.  Again, this use of backward evolution is atypical of 
conventional approaches to adiabatic evolution where the additional evolution time/speed required by backwards evolution would tend
to be detrimental.  In contrast, such backwards evolutions can lead to substantial reductions in the cost for coherently controlled adiabatic
evolution.

\subsection{Interpolation}\label{sec:gencond}
There are many ways that these requirements can be satisfied by a dual path to $f_A$.  The way that we satisfy these requirements is by adding a
smooth polynomial continuation of $f_A$ about $s=1$ that allows the derivative to loop around and attain the opposite value.  This interpolation must
have piecewise continuous third derivatives in order to guarantee that the $O(1/T^2)$ terms will remain sub--dominant in the limit of large $T$.  This
naturally leads to a quartic interpolation that takes the following form for a partially anti--symmetric combination
\begin{equation}
f_B(s)=\begin{cases} f_A(s) & s<1-\Delta \\
es^4+ds^3+cs^2+bs+a & s\ge 1-\Delta \end{cases},\label{eq:onedogeara}
\end{equation}
where $\Delta$ is a free parameter that controls when $f_B$ switches from the original 
adiabatic path $f_A$ to the polynomial interpolation.  The parameters are then set by 
requiring 
\begin{align}
f_B(1) &= f_A(1)=1\nonumber\\
f_B(1-\Delta) &= f_A(1-\Delta)\nonumber\\
\dot f_B(1) &= - \dot f_A(1)\nonumber\\
\dot f_B(1-\Delta) &= \dot f_A(1-\Delta)\nonumber\\
\ddot f_B(1-\Delta) &= \ddot f_A(1-\Delta)\label{eq:onedogearb}
\end{align}
We could also have equally well choose the backwards evolution to start at $s=0$ rather than $s=1$.  Although seemingly arbitrary, this choice can
have a substantial impact on the error depending on whether the gap is larger at $s=0$ and $s=1$.  We also make use of this fact later
in~\sec{general} where we exploit this fact to suppress every transition simultaneously for Hamiltonians that satisfy a certain symmetry property.

The case of fully anti--symmetric boundaries is similar except now two polynomial interpolations are needed:
\begin{equation}
f_B(s)=\begin{cases} f_A(s) & \Delta/2<s<1-\Delta/2 \\
es^4+ds^3+cs^2+bs+a & s\le \Delta/2\\
e's^4+d's^3+c's^2+b's+a' & s\ge 1-\Delta/2 \end{cases},\label{eq:twodogeara}
\end{equation}
where $\Delta$ is a free parameter that controls how rapidly $f_B$ switches from the original adiabatic path, $f_A$, to the polynomial interpolation. 
The parameters are then set by requiring 
\begin{align}
f_B(0) &= f_A(0)=0\nonumber\\
f_B(1) &= f_A(1)=1\nonumber\\
f_B(1-\Delta/2) &= f_A(1-\Delta/2)\nonumber\\
f_B(\Delta/2) &= f_A(\Delta/2)\nonumber\\
\dot f_B(1) &= - \dot f_A(1)\nonumber\\
\dot f_B(0) &= - \dot f_A(0)\nonumber\\
\dot f_B(1-\Delta/2) &= \dot f_A(1-\Delta/2)\nonumber\\
\dot f_B(\Delta/2) &= \dot f_A(\Delta/2)\nonumber\\
\ddot f_B(\Delta/2) &= \ddot f_A(\Delta/2)\nonumber\\
\ddot f_B(1-\Delta/2) &= \ddot f_A(1-\Delta/2).\label{eq:twodogearb}
\end{align}
In particular, the coefficients in~\eq{twodogeara} can then be found by substituting~\eq{twodogeara} into~\eq{twodogearb}.

It then follows that for any fixed path $f_A$, we can choose $f_B$ such that the dominant transition is suppressed to $O(1/T^2)$.  This opens the
possibility that our error suppression methods may allow adiabatic state preparation to be performed using less evolution time (or equivalently, fewer
gates on a quantum computer) than existing methods.  However, local adiabatic evolution is known to be optimal for performing adiabatic Grover's
search~\cite{VDM+01,RC02}, so we cannot expect that the algorithm will outperform all existing adiabatic algorithms in every time regime.  We will see
below, that although our method does not outperform local adiabatic evolution for short times it can come very close to matching its performance while
giving substantially reduced error for slow evolutions.

\section{Comparison to Local Adiabatic Evolution}\label{sec:LA}
We focus our numerical results on the case of adiabatic Grover's search.  The Hamiltonian for adiabatic Grover's search is
\begin{equation}
H(f(s)) = (1-f(s)) (\openone - \ketbra{+^{\otimes n}}{+^{\otimes n}}) + f(s)(\openone - \ketbra{0}{0}),
\end{equation}
and sufficiently slow evolution of this Hamiltonian causes the initial eigenstate $\ket{+^{\otimes n}}$ to transition to the marked state $\ket{0}$ as
per Grover's search.
Local adiabatic evolution is known to be optimal for adiabatic Grover's search~\cite{RC02,VDM+01} meaning that the quadratic speedup over classical
algorithms
is attained for the adiabatic path.

\begin{figure}[t]
\centering
\includegraphics[width=30em]{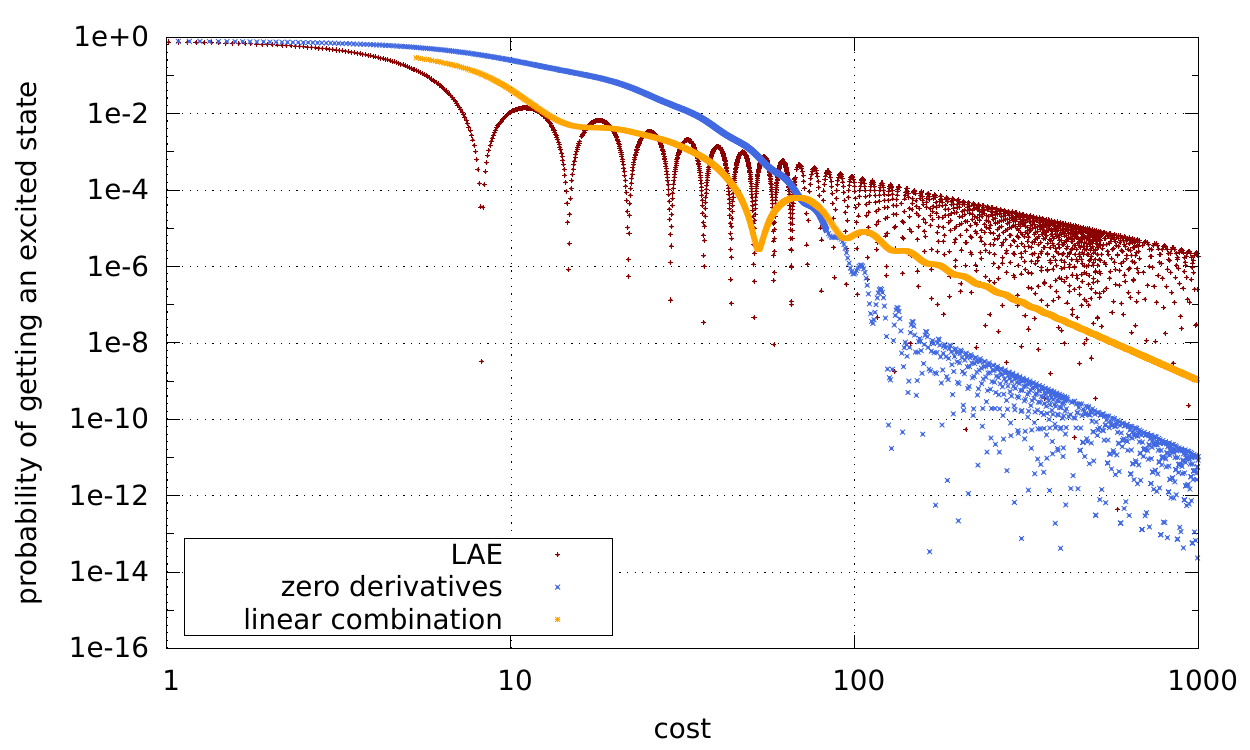}
\caption{Diabatic errors for
local adiabatic evolution, boundary cancellation with one zero--derivative and linear
combination of the local adiabatic evolution and an evolution with the opposite
derivative at the end for the Search Hamiltonian with $N=5$.}
\label{fig:partialsymmetric}
\end{figure}

The path for local adiabatic evolution, for cases where the search space is $N$--dimensional, is 
\begin{equation}
f(s)=\frac{\sqrt{N-1}-\tan\big[\arctan(\sqrt{N-1})(1-2s)\big]}{2\sqrt{N-1}}\label{eq:localevolution}.
\end{equation}
Our goal is to compare the cost of performing this adiabatic quantum algorithm using local adiabatic evolution to the cost incurred by using our
methods.  We choose $f_A$ to be the path given by local adiabatic evolution whereas $f_B$ is taken to be a continuation of the local adiabatic path
that satisfies~\eq{fbderivs} or~\eq{completeSymmetry} in all of the following numerical examples.

There are several ways that the cost of an adiabatic algorithm can be measured.  The most straight forward method is to compare the time required for
the evolution.  Although this cost metric is appropriate in cases where the norm of the Hamiltonian is fixed, it is not appropriate for comparing
different adiabatic evolutions because the energy required for both paths may differ substantially.  Since there is a duality between energy and time
in quantum mechanics, a fast evolution that requires a lot of energy may be dynamically equivalent to a slow evolution that requires little energy. 
Thus we need to consider not just the time but also the energy.  For this reason, we use the following cost metric for the case 
where we combine $j$ evolutions (where $j\ge 1$):
\begin{equation}
{\rm Cost} = \max_j \left\{\int_0^1 \|H_j(s)\| \mathrm{d} s T_j\right\}\biggr/ P(0),
\end{equation}
where $P(0)$ is the success probability of the gadget which is given by~\eq{pfail}.
Here we implicitly assume that the cost of the rotations and the control logic is negligible and that each of the evolutions can be implemented in
parallel.  These assumptions may not hold in general, but they are appropriate for quantum computer simulations of such adiabatic evolutions because
the query complexity of such evolutions depends on the maximum evolution time chosen rather than the total evolution time.  This point will be made
clear in~\sec{qcimp}.

We see in~\fig{partialsymmetric} that including the second adiabatic path with partially anti--symmetric boundary conditions (as per~\sec{partial}) to
LAE yields comparable performance to LAE for short evolutions and also provides the improved scaling of boundary cancellation methods in the adiabatic
regime.  In particular, the second path follows the interpolation strategy of~\eq{onedogeara}: it follows LAE (i.e.~\eq{localevolution}) until $s=0.8$
and then smoothly transitions to a fourth--order polynomial.  Unlike the method of~\cite{nathans}, this approach does not only give superior error
scaling over a discrete set of points; although, this technique does enforce a minimum evolution time as discussed in~\sec{partial}.  
An important drawback of this method is that there is a manifest lack of symmetry in the derivatives in the adiabatic regime for this method.  This
means that the adiabatic interference effects that appear in the LAE and boundary cancellation paths will not appear here~\cite{nathans}.  Note that
if the Search Hamiltonian did not have symmetric derivatives or spectrum, the adiabatic interference effects would not appear and so they are an
artifact of having a highly structured test Hamiltonian.

\fig{fullsymmetric} tells a similar story.  In that case we use fully anti--symmetric boundary conditions and add a second path that interpolates
between LAE and polynomial evolution as per~\eq{twodogeara} with $\Delta=0.2$.  This also corresponds to evolution under LAE for $80\%$ of the time. 
Unlike the case in~\fig{partialsymmetric} adiabatic interference patterns are again visible in the adiabatic regime because the two polynomials used
to create the fully anti--symmetric boundary conditions between the two paths at $s=0$ and $s=1$ ensure that the derivatives are the same at the
boundary; thereby allowing such interference effects to emerge again which causes the error to be substantially reduced on a discrete subset of points
as per~\cite{nathans}.  As a consequence we can clearly see that our method substantially outperforms both methods for a range of evolutions with cost
ranging from $[50,100]$ due in part to the presence of adiabatic interference effects that are absent from the boundary 
cancellation method.

In both of the cases considered, our methods are less effective at suppressing errors in the adiabatic regime than boundary cancellation methods. 
This is because the $O(1/T^2)$ terms in the error in the adiabatic approximation also depend on $\dot{H}(s)$.  Such terms are zero for boundary
cancellation methods and so we generically expect from the triangle inequality that boundary cancellation will lead to less error in this regime.  
An important point to note is that although these test cases do not outperform LAE for fast evolutions or boundary cancellation methods for slow
evolutions, they can outperform both methods for evolutions that operate at an intermediate speed.  This implies that these methods are not just a
compromise between the two approaches: they also provide superior scaling in a region that is badly addressed by existing adiabatic optimization
methods.

\begin{figure}[t]
\centering
\includegraphics[width=30em]{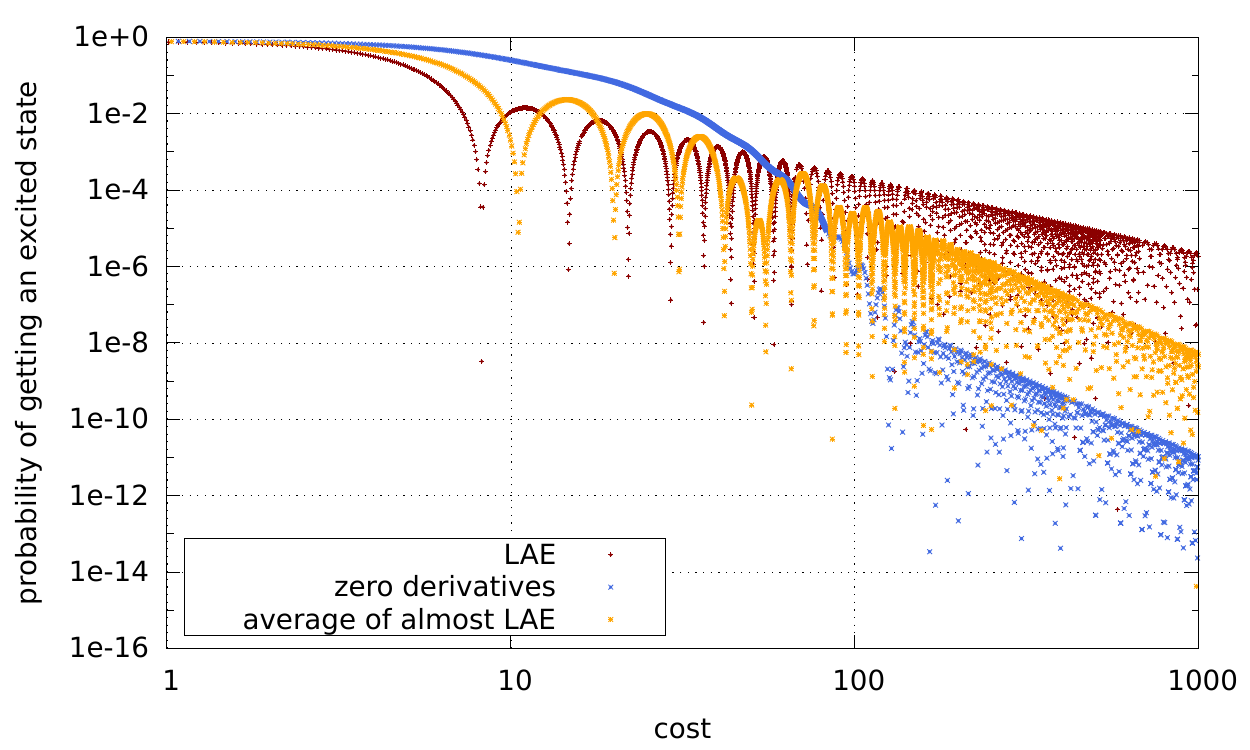}
\caption{Diabatic errors for
local adiabatic evolution, boundary cancellation with one--zero derivative and linear
combination of the two paths that resemble the local adiabatic evolution for most
of the time but have opposite derivatives the
beginning or the end
respectively for the Search Hamiltonian with $N=5$.}
\label{fig:fullsymmetric}
\end{figure}

\section{Suppressing every transition for symmetric $H$\label{sec:general}}

We now consider suppressing errors for Hamiltonians where $H(0)$ and $H(1)$ have the same spectra.
Although restrictive, this condition is satisfied in many natural problems~\cite{FG00,RC02,nathans}. 
This symmetry is very useful because it guarantees that two adiabatic interpolations exist between $H_0$ and $H_1$
such that the amplitudes for every state orthogonal to $\ket{g(1)}$ that arise due to $f_A$ are equal and opposite to those
that arise under $f_B$.  This means that the linear combination will simultaneously suppress diabatic leakage into every state.
 In contrast, the methods discussed in~\sec{partial} and~\sec{complete}
only guarantee this for a single (but arbitrarily chosen) transition.

First let us assume that the following conditions are met for $f_A(s)$ and $f_B(s)$
\begin{align}
\int_0^s \gamma_{g,n}^A\left(f_A\left(\xi\right)\right) \mathrm{d}\xi &= \int_0^s
\gamma_{g,n}^B\left(f_B\left(\xi\right)\right) \mathrm{d}\xi\label{eq:gaprequire}\\
\left.\frac{\braket{\dot{n}(f_A(s))}{m(f_A(s))}_A}{\gamma_{g,n}^A\left(f_A(s)\right)}\right|_{s=0}&=
-\left.\frac{\braket{\dot{n}(f_B(s))}{g(f_B(s))}_B}{\gamma_{g,n}^B\left(f_B(s)\right)}\right|_
{s=0}\label{eq:gencond1}\\
\left.\frac{\braket{\dot{n}(f_A(s))}{m(f_A(s))}_A}{\gamma_{g,n}^A\left(f_A(s)\right)}\right|_
{s=1}&=
-\left.\frac{\braket{\dot{n}(f_B(s))}{g(f_B(s))}_B}{\gamma_{g,n}^B\left(f_B(s)\right)}\right|_
{s=1}\label{eq:gencond2}
\end{align}
for all states $\ket{n}\ne \ket{g}$.  We will see that these conditions can always be met if the spectrum of $H(s)$ is symmetric about $s=1/2$.

Such conditions do not naturally arise for all adiabatic passages but there are many examples where such Hamiltonians are natural.  A natural example
is the
Search Hamiltonian; however, such an application is trivial because the quantum dynamics occurs within a two--dimensional subspace.  Other examples
occur in adiabatic
gates~\cite{Hen14,AMP+07,SBR+08} and holonomic quantum computing~\cite{ZR99,DCZ01}.

After substituting \eq{gaprequire},~\eq{gencond1} and~\eq{gencond2} into~\eq{errorxpr} we find

\begin{figure}
\centering
\includegraphics[width=27em]{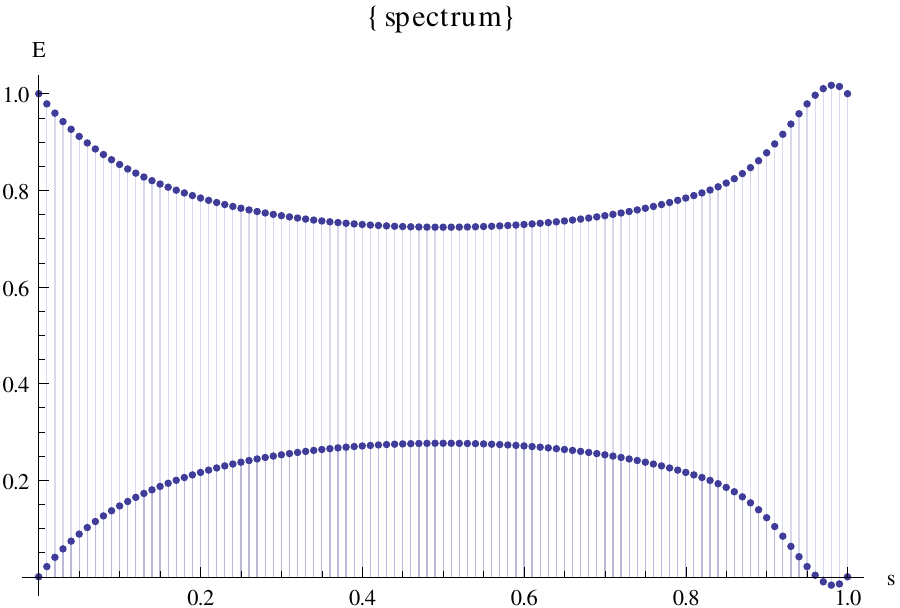}\
\includegraphics[width=27em]{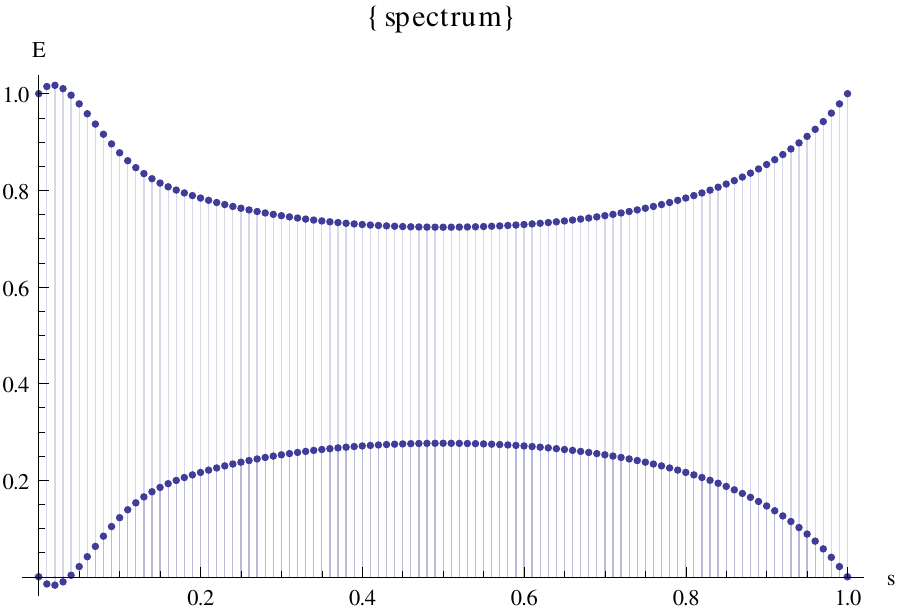} 
\caption{A symmetric pair of Hamiltonians with equal gap integrals and opposite derivatives at the 
begining and at the end. The same intuition formalized in \eq{gencond1} works for an 
arbitrary number of transitions.}
\label{fig:symmetric}
\end{figure}

\begin{align}
|\ketbra{n(1)}{n(1)}\ket{\phi}|=&\Biggr|\cos^2{(\theta)}
\left[\frac{\braket{\dot{n}^A_1(1)}{g(1)}}{-i\gamma^A_{g,n}\left(1\right)T_A}
-e^{+i\int_0^1 \gamma^A_{g,n^A}\left(\xi\right)T_A d \xi}
\frac{\braket{\dot{n}^A(0)}{g(0)} 
}{-i\gamma^A_{g,n}\left(0\right)T_A}\right]\nonumber \\
&-\sin^2{(\theta)}
\left[\frac{\braket{\dot{n}^A(1)}{g(1)}}{-i\gamma^A_{g,n}\left(1\right)T_A} - 
e^{+i\int_0^1 \gamma^A_{g,n}\left(\xi\right)T_A d \xi}  
\frac{\braket{\dot{n}^A(0)}{g(0)}}{-i\gamma^A_{g,n}\left(0\right)T_B}\right]\Biggr|+O(1/T^2).\label{eq:diabaticerror}
\end{align}
It is then clear that if we take $\theta=\pi/4$ and $T_A=T_B$ then every transition will be suppressed from $O(1/T)$ to $O(1/T^2)$ under these
assumptions.

The question remaining is: when can we make these conditions hold?  A natural case that covers a 
wide range of adiabatic protocols is the case where the eigenvalue gap is symmetric.  That is to say 
that $\gamma_{g,n}(s) = \gamma_{g,n}(1-s)$ for all $s$ and $n\ne 
g$.  It is difficult to find a second adiabatic path that satisfies the conditions 
in~\sec{gencond}, for the choice $f_A=f$ because the  anti--symmetry required by~\eq{gencond1}, 
\eq{gencond1} necessitates the use of adiabatic paths similar to those in~\sec{complete}.  Such 
paths will typically violate~\eq{gaprequire} because including the reversal near $s=0$ and $s=1$ 
will change the gap integral.  

A better approach is to modify \emph{both} paths.  It is easy to see by substitution that if we let $f_B$ be given by~\eq{onedogeara}
and~\eq{onedogearb} and then choose $f_A$ to be the time reversed version of this path (i.e. $f_A(s)=1-f_B(1-s)$).  It is then easy to see by substitution that under there assumptions
\begin{align}
H(f_A(s)) &=(1-f_A(s))H_0+f_A(s)H_1\nonumber\\ 
&=f_B(1-s)H_0+(1-f_B(1-s))H_1.
\end{align}
The assumption that $\gamma_{g,n}(s)=\gamma_{g,n}(1-s)$ then directly implies that $\gamma_{g,n}(f_A(s))= \gamma_{g,n}(f_B(1-s))$, which gives us the
desired result of
\begin{align}
\int_0^1 \gamma_{g,n}(f_B(s)) \mathrm{d} s=\int_{0}^{1} \gamma_{g,n}(f_A(s))\mathrm{d} s.\label{eq:gapintegrals}
\end{align}
This fact becomes immediately obvious in light of~\fig{symmetric}, where the spectrum for a Search Hamiltonian with $f_A$ and $f_B$ chosen to be time
reverses of each other is given.
After substituting~\eq{gapintegrals} into~\eq{diabaticerror} and using the assumption that $\gamma_{g,n}(1)=\gamma_{g,n}(0)$ we see that
\begin{equation}
\ketbra{n(1)}{n(1)}\ket{\phi}= O(1/T^2).
\end{equation}
Thus for any adiabatic path parameterized by $f(s)$ and any evolution time $T$, we can always choose two paths $f_A$ and $f_B$ that both incur
diabatic errors that cancel to $O(1/T^2)$.
This fact is demonstrated numerically for a Hamiltonian that satisfies these requirements in~\fig{more}.  In contrast, without these assumptions, only
the dominant transition is suppressed to $O(1/T^2)$ which implies that the diabatic errors scale as $O(1/T)$ for sufficiently long evolutions.
\begin{figure}
\centering
\includegraphics[width=27em]{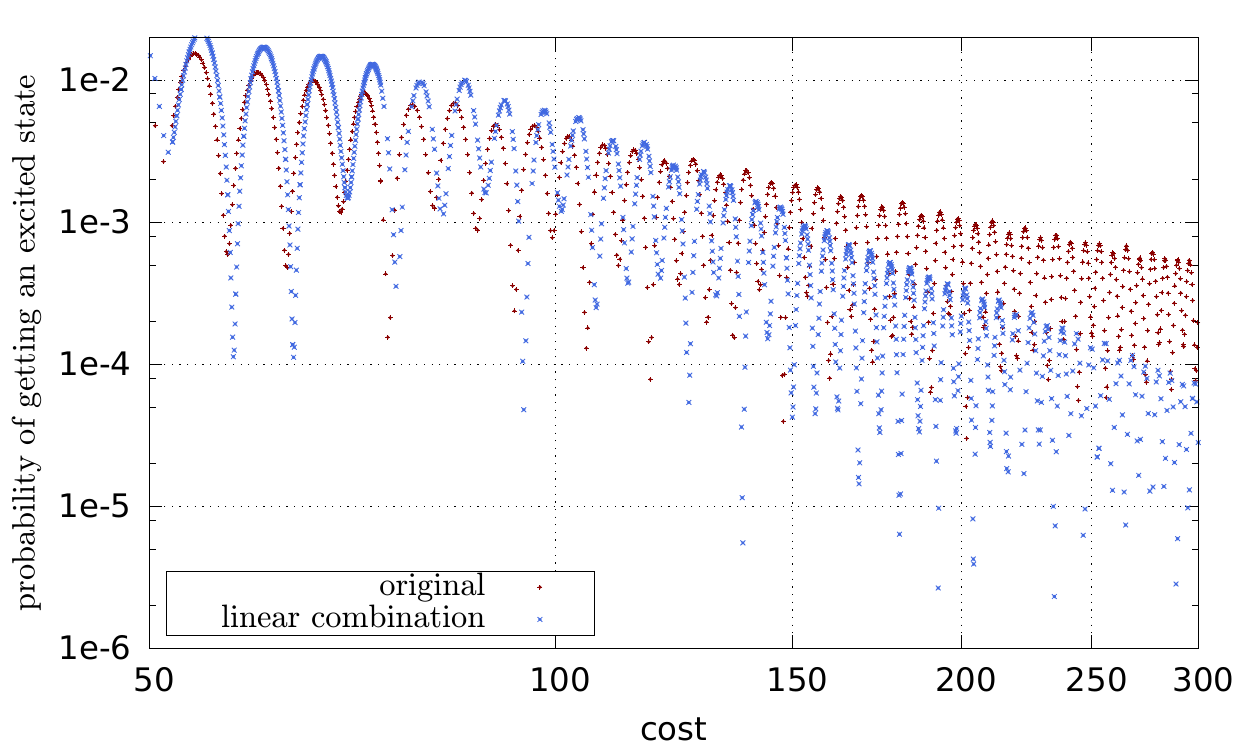}
\caption{Diabatic errors for $H(s) = \sigma_z^{(1)}+\sigma_z^{(2)}+\sin(\pi s) \mathcal{H}\otimes\mathcal{H}$ where $\mathcal{H}$  is the Hadamard
operator, for $H(s)$ directly and also for the case where the adiabatic paths are chosen as per~\eq{gencond1} and~\eq{gencond2} with $\Delta=0.1$ for
the linear combination.}
\label{fig:more}
\end{figure}

This result is much stronger than that of~\cite{nathans}, which only leads to suppression of~\emph{all} diabatic errors if the gap integrals for each
transition are rational multiples of each other and, even then, will only work at specially chosen values of $T$.  This precludes the technique's use
for almost all Hamiltonians.  In contrast, coherent control of the adiabatic path allows all of the transitions to be suppressed for a wide class of
adiabatic protocols and this result holds for any $T$.  This clearly shows demonstrates that coherently controlled adiabatic evolution allows us to
circumvent the limitations of existing adiabatic optimization schemes.  We will also see this below, where we show that these methods can be used in
concert with boundary cancellation methods.

\section{Incorporating Boundary Cancellation}\label{sec:imph}
The phase cancellation method allows one to reduce the error in adiabatic evolution to 
$O\left(\frac{1}{T^{m+1}}\right)$ by setting first $m$ derivatives of the Hamiltonian zero at the 
boundaries as shown in~\cite{LRH09,RPL10}. Our approach can incorporate these results in a natural way. 
When we set first $m$ derivatives to zero on the boundary, the remaining error is~\cite{nathans}
\begin{equation}
\left.
\left(\ii-\ket{g\left(1\right)}\!\bra{g\left(1\right)}\right)U\left(1,
0\right)\ket{g\left( s \right) }=
\sum_{n\neq g}e^{-i\int_0^1 E_n\left(\xi\right) T d
\xi}\frac{\bra{ n(s)}(\partial_s^{m+1}H(s))\ket{g(s)}e^{i\int_0^s\gamma_{g,n}\!(\xi)d\xi
T}}{-i\gamma_{g,n}^{m+2}(s)T^{m+1}}\right|_{s=0}^1
\ket{n\left(1\right)} +O\left(\frac{1}{T^{m+2}}\right) \label{eq:m_error}.
\end{equation}
This expression is analogous to \eq{error}, as can be seen by substituting
\begin{align}
\braket{\dot{n}(s)}{m(s)} &\rightarrow \frac{\bra{n(s)}(\partial_s^{m+1} H(s))\ket{m(s)}}{\gamma_{g,n}^{m+1}(s)}\nonumber,\\
T&\rightarrow T^{m+1}.\label{eq:subs}
\end{align}
into~\eq{error}.
Equation~\eq{m_error} implicitly assumes that $\braket{\dot n(s)}{n(s)}=0$, which can always be assumed to be true because the phases of $\ket{n(s)}$
are arbitrary.
The previous methods can then be used after making these substitutions and using a higher--order polynomial to perform the interpolation.  For
example, we can generalize the method of~\sec{general}
with the following modification to the conditions required for both $f_A$ and $f_B$: 
\begin{align}
f_B^{(m+1)}(0) &= -  f_A^{(m+1)}(0)\nonumber\\
f_B^{(m+1)}(1) &= -  f_A^{(m+1)}(1)\nonumber\\
f_B^{(j)}(\Delta) &=  f_A^{(j)}(\Delta)~\forall~j\in [0,m+2]\nonumber\\
f_B^{(j)}(1-\Delta) &=  f_A^{(j)}(1-\Delta)~\forall~j\in [0,m+2]\nonumber\\
f_B^{(j)}(1) &=  0~\forall~j\in (0,m],\label{eq:twodogearc}
\end{align}
and picking $T_A=T_B$ with $\theta=\pi/4$.  This enables exponentially accurate adiabatic approximations if $m$ is chosen as a function of $T$.

Alternatively, cancellation of the~\emph{leading order} transition to $O(1/T^{m+1})$ can be obtained by using the exact same ideas within the methods
of~\sec{partial} and~\sec{complete}
after using the substitutions in~\eq{subs} and conditions similar to~\eq{twodogearc} for the polynomial interpolations.
It should also be noted that in systems where the adiabatically transported subspace is a one-- or two--dimensional subspace then such approaches also
can be used to make the overall error scaling $O(1/T^{m+1})$.  This raises the possibility that controlled adiabatic evolution can be combined with
boundary cancellation methods to substantially reduce the cost of performing a high--accuracy adiabatic state preparation.

\section{Costing Controlled Adiabatic Evolutions\label{sec:qcimp}}
There are two types of costly resources in coherently controlled adiabatic evolutions.  The first is the cost of evolving the adiabatic register,
denoted $\ket{\psi}$ in~\fig{circuit}.  The second is the cost of performing the required rotations on the control register.  In this section we will
provide a complete cost analysis of this model under the assumption that it is being simulated using a circuit based quantum computer that is further
equipped with oracles that compute the necessary properties of the Hamiltonian.  We will then conclude that coherently controlled adiabatic using
sparse, row--computable Hamiltonians evolution is polynomially equivalent to the circuit model. Other appropriate cost models, such as bounding the
energy and time required to implement the controlled Hamiltonians using $k$--local Hamiltonians will not be discussed here.  Since we show that
coherently controlled adiabatic evolution is polynomially 
equivalent to the circuit model, it will immediately follow that it is also polynomially equivalent to adiabatic quantum computation using local
Hamiltonians.

The first important result that we need to show this is an upper bound on the number of oracle queries needed to simulate a time--dependent
Hamiltonian within fixed error on a quantum computer.  We will use this result to upper bound the query complexity of performing the controlled
adiabatic evolutions.  In order to understand the theorem, we will define a smoothness classification for Hamiltonians:
\begin{definition}\label{def:lambdasuzsmooth}
The set of operators $\{H_j:j=1,\dots, m\}$ is $\Lambda$-$P$-smooth on~$\mathcal{I}\subseteq \mathbb{R}$ if 
$\Lambda\geq\left(\sum_{j=1}^m\left\|\partial_t^p H_j(t)\right\|\right)^{1/(p+1)}$,
for all $t\in\mathcal{I}$ and~$p\in\{0,1,\cdots,P\}$.
\end{definition}
Now with~\defin{lambdasuzsmooth} we can state the following corollary, which gives the query complexity of the simulaiton algorithm.  The number of
one-- and two--qubit needed for the simulation is at most proportional to the number of queries made.

\begin{corollary}[Cor. 6 of~\cite{WBHS11}]
\label{cor:piecewise}
Let $\{H_\alphavar:\mathbb{R}\mapsto\mathbb{C}^{2^n\times 2^n};\alphavar=1,\ldots,M\}$ be a set of time-dependent Hermitian operators that is
$\Lambdavar$-$2k$-smooth on~$\mathcal{I}=(\initt,\initt+\dt)\setminus\{t_1,\ldots, t_L\}$, where~$\initt<t_1<\cdots<t_L<\initt+\dt$, with the
additional conditions
\begin{enumerate}
\item $\exists~H_{\rm{max}}\in \mathbb{R}: H_{\max}\ge\max_{t\in[\initt,\initt+\dt]}\|H(t)\|$,
\item $0<\epsilon\le\min\{1, 27(5/3)^{k-1}d^2\Lambdavar\dt\}$,
\item $\numBitsInT$ satisfies $\numBitsInT\ge \left\lceil\log_2\left(\frac{(\max_{t\in
\mathcal{I},\alphavar}\|\partial_tH_{\alphavar}(t)\|)(32kMd^2)(5/3)^{k-1}\dt^2}{\epsilon}\right)\right\rceil$, 
\item $\numQubitsInH$ satisfies $\numQubitsInH\ge 2\left\lceil\log_2\left(\frac{32k M d^2(5/3)^{k-1}\Lambdavar\dt}{\epsilon}\right)\right\rceil+6$ and
\item $\dt/2^{\numBitsInT}< \min_{\ell=0,\ldots,L}(t_{\ell+1}-t_\ell)$ with~$t_{L+1}:=\initt+\dt$,
\end{enumerate}
where $\numBitsInT$ is the number of bits used to represent $t$ and $\numQubitsInH$ is the number of qubits used to encode the matrix elements of $H$.
Then the query complexity for simulating evolution generated by $H(t)=\sum_\alphavar T_\alphavar^{\dagger}H_\alphavar(t) T_\alphavar$, for fixed set
of unitary basis changing operators $T_\alphavar$, within an error of~$\epsilon$ using time--ordered Trotter--Suzuki formulas with error
$O(\dt^{2k+1})$ is
\begin{align}
\label{eq:cor3}
\Noracle&\le 12 \Costvar Md^2 5^{k-1}\left[(L+1)+ 24kd^2\Lambdavar \dt\left(\frac{5}{3}
\right)^{k}\left(\frac{6d^2\Lambdavar\dt}{(\epsilon/3)}\right)^{1/2k}\right],
\end{align}
where $\Costvar$ is the number of oracle calls needed to simulate a one-sparse Hamiltonian, and the number of basis change operations is at most
$\Noracle/(3\Costvar d^2)$.
\end{corollary}

Note that we need to use a result for simulating piecewise smooth Hamiltonians because the interpolations used in our method will typically cause
$H(f_j(s))$ to be non--analytic at either one or two points.  The result of~\cor{piecewise} is then useful because it provides the cost of performing
such a simulation despite such complications.  Note that in the cases we consider $L=1$ for partially anti--symmetric boundary conditions and $L=2$
for completely anti--symmetric boundary conditions.  Also, for simplicity we cite a method that does not use adaptively chosen timesteps.  Such
adaptive methods are given in~\cite{WBHS11} and lead to similar scaling where $\Lambda$ is replaced by the time average of the instantaneous values of
$\Lambda$.

There are two types of oracles that are required by this corollary.  Firstly an oracle is required that outputs the location of the $j^{\rm th}$
non--zero matrix element in a given row, where $j\le d$ if $H$ is \sparse{d}.
\begin{equation}
O_1(j):\ \ket{i}\ket{0}\rightarrow \ket{i}\ket{L_j(i)}
\end{equation}
where $L_j(i)$ gives the $j^{\rm th}$ non-zero element in row $i$.  We cost a single query to $O_1$ as $n$ queries to a single qubit oracle because
each query to this oracle yields $n$ bits, and it is more realistic to cost the algorithm by the number of qubits output if the dimension of the
Hilbert space is large.  The corollary also requires an oracle
for matrix elements of $H(s)$
\begin{equation}
O_2:\ \ket{i}\ket{k}\ket{s}\ket{0}\rightarrow
\ket{i}\ket{k}\ket{s}\ket{H(s)_{i,k}}.
\end{equation}

We also use (for convenience) a new oracle, $O_f$, whose role is to prepare a quantum state encoding the particular value of $f_A(s)$ or $f_B(s)$ that
is needed in a given timestep.  In general, if we wish to find the value of $f_p(s)$ we use the oracle in the following way:
\begin{equation}
O_f(s) \ket{p}\ket{0} = \ket{j}\ket{f_p(s)}.
\end{equation}
This oracle is crucial to our approach because it allows us to remove the multiple controls used in~\fig{circuit}.  For example,
\begin{equation}
O_2O_1O_f(s) \sum_{p=1}^N \sum_{x=1}^{2^n} b_p a_{j} \ket{p}\ket{x} \ket{0}\ket{0}\ket{0} = \sum_{p=1}^N \sum_{x=1}^{2^n} b_p a_{j} \ket{p}\ket{x}
\ket{L_{\ell}(x)}\ket{f_p(s)}\ket{H(f_p(s))_{x,L(x,i)}}
\end{equation}

Our cost analysis of the controlled adiabatic evolution follows by converting the controlled evolution in~\fig{circuit} into the evolution of a single
larger Hamiltonian.  This larger Hamiltonian can then be simulated by conventional means (such as a Trotter--Suzuki decomposition as per~\cor{piecewise}).
 
\begin{theorem}\label{thm:efficient}
Assume that we wish to simulate  a coherently controlled adiabatic evolution that uses the controlled evolutions $\{\mathcal{T}e^{-i\int_0^1
H(f_1(s))T_1},\ldots,\mathcal{T}e^{-i\int_0^1 H(f_p(s))T_p}\}$ such that $H(s)$ is a Hamiltonian satisfying $\|H'(s)\| \le \Gamma$ and each $H(f_j)$
is $\Lambda$--$2k$--smooth and all remaining assumptions of~\cor{piecewise} are held for $\Delta t= \max_j T_j$.  The query complexity of performing
the simulation within error at most $\epsilon$ using $k^{\rm th}$--order time--ordered Trotter--Suzuki formulas and oracles that yield one bit per
query obeys
\begin{align}
\label{eq:cor4}
\Noracle&\le 12 \Costvar Md^2 5^{k-1}\left[(L+1)+ 24kd^2\Lambdavar \max_j T_j\left(\frac{5}{3} \right)^{k}\left(\frac{6d^2\Lambdavar\max_j
T_j}{(\epsilon/6)}\right)^{1/2k}\right],
\end{align}
where
\begin{equation}
C\le 4n(z_n +2) + 3n_H + 2\left\lceil\log_2 \left(\frac{6\Gamma \max_j T_j}{\epsilon} \right) \right\rceil,
\end{equation}
 $z_n$ is the number of times that $n\mapsto \lceil 2\log_2 n \rceil$ must be iterated before achieving a value that is less than or equal to $6$ and
$N_{f(s)}\le \lceil \log_2(\Gamma)\rceil$.
\end{theorem}

\begin{proof}
To see this, note that the controlled unitary evolutions in~\fig{circuit} produce a time--evolution operator of the following block--diagonal form:
\[
     \begin{bmatrix}
    U_1 &\cdots & 0\\
\vdots & \ddots &\vdots\\
    0 &\cdots& \openone\\
  \end{bmatrix}\times \cdots \times
  \begin{bmatrix}
    \openone &\cdots & 0\\
\vdots & \ddots &\vdots\\
    0 &\cdots& U_p\\
  \end{bmatrix}=  \begin{bmatrix}
    U_1 &\cdots& 0\\
\vdots &\ddots &\vdots\\
    0 &\cdots& U_p\\
  \end{bmatrix},
\]
where $U_1,\ldots,U_p$ are the $p$ controlled adiabatic evolutions.  By expanding out the unitaries as time--ordered operator exponentials, we see
that the ideal time evolution operator is of the form
\begin{equation}
 \begin{bmatrix}
    U_1 &\cdots& 0\\
\vdots &\ddots &\vdots\\
    0 &\cdots& U_p\\
  \end{bmatrix}= \begin{bmatrix}
    \mathcal{T} e^{-i\int_0^1 H(f_1(s)) \mathrm{d}s T_1} &\cdots& 0\\
\vdots &\ddots &\vdots\\
    0 &\cdots& \mathcal{T} e^{-i\int_0^1 H(f_p(s)) \mathrm{d}s T_p} \\
  \end{bmatrix}. \label{eq:controlled}
\end{equation}

Consider the Hamiltonian $H= \sum_j \ketbra{j}{j} \otimes H_j$.  It is easy to see using Taylor expansion and the fact that each of the terms in $H$
commute that
\begin{align}
e^{-iHt} &= e^{-i\sum_j \ketbra{j}{j}\otimes  H_jt}=\prod_{j=1}^p e^{-i \ketbra{j}{j} \otimes H_j t}\nonumber\\
&= \prod_{j=1}^p (\ketbra{j}{j}\otimes e^{-i H_j t}+(\openone-\ketbra{j}{j})\otimes \openone)\nonumber\\
&= \bigoplus_{j=1}^p \ketbra{j}{j}\otimes e^{-i H_j t}.
\end{align}
Here $\bigoplus$ represents the direct sum operation.  Thus we have that
\begin{equation}
e^{-iHt} = \begin{bmatrix}
     e^{-iH_1 t} &\cdots& 0\\
\vdots &\ddots &\vdots\\
    0 &\cdots&  e^{-iH_p t} \\
  \end{bmatrix}.\label{eq:blockdiag}
\end{equation}

Now let $\mathbf{H}(s)= \sum_j \ketbra{j}{j} \otimes H_j(s)$.  It then follows from the definition of the ordered--operator exponential and the
block--diagonal structure of~\eq{blockdiag} that
\begin{align}
\mathcal{T} e^{-i\int_0^1 H(s) \mathrm{d}s T}&= \lim_{r\rightarrow \infty} \prod_{j=1}^r \begin{bmatrix}
     e^{-iH_1(j/r) T/r} &\cdots& 0\\
\vdots &\ddots &\vdots\\
    0 &\cdots&  e^{-iH_p(j/r) T/r}
\end{bmatrix}\nonumber\\
&= \begin{bmatrix}
    \mathcal{T} e^{-i\int_0^1 H_1(s) \mathrm{d}s T} &\cdots& 0\\
\vdots &\ddots &\vdots\\
    0 &\cdots& \mathcal{T} e^{-i\int_0^1 H_p(s) \mathrm{d}s T} \\
  \end{bmatrix}. \label{eq:blockdiag2}
\end{align}
It then follows that the controlled evolutions in~\eq{controlled} can be expressed as a simulation of a single time--dependent Hamiltonian by taking
$H_j(s) \rightarrow H(f_j(s)) T_j /T$ in~\eq{blockdiag2}.  For simplicity, let us take $T=\max_j T_j$.

Next we need to find the properties of the dilated Hamiltonian $H(s)$ that describes the controlled evolutions in the controlled adiabatic evolution. 
Firstly, assuming each $H_j$ is the sum of $M$  Hamiltonians that can be efficiently transformed into \sparse{d} matrices, it follows that
$\mathbf{H}$ can be expressed as a similar sum.  Similarly, since $H(s) = \sum_j \ketbra{j}{j}\otimes H(f_j(s))T_j/T$, it follows from the fact that
$\mathbf{H}(s)$ has a direct sum structure and the assumption that each $H(f_j(s))$ is $\Lambda$--$2k$--smooth that for any non--negative integer
$q\le 2k$
\begin{equation}
\| \partial_s^q \mathbf{H}(s) \| = \max_j \|\partial_s^q H(f_j(s))\|\frac{T_j}{T} \le \Lambda^{q+1}.
\end{equation}
Hence, for $T=\max_j T_j$, $H(s)$ is at most $\Lambda$--$2k$--smooth.

We then have from~\eq{cor3} that the cost of simulating the effective Hamiltonian $\mathbf{H}(s)$ using $k^{\rm th}$--order Trotter--Suzuki formulas
is at most
\begin{align}
\label{eq:cor5}
\Noracle&\le 12 \Costvar Md^2 5^{k-1}\left[(L+1)+ 24kd^2\Lambdavar \max_j T_j\left(\frac{5}{3} \right)^{k}\left(\frac{6d^2\Lambdavar\max_j
T_j}{(\epsilon/3)}\right)^{1/2k}\right].
\end{align}

The remaining issue is the calculation of $C$.  In order to compute $C$ we need to first show that we can simulate a query to the Hamiltonian oracles
for $\mathbf{H}$ using those for $H$.  We specifically require two oracles: one that computes the locations of the $i^{\rm th}$ (potentially)
non--zero matrix element in any row of $\mathbf{H}$ and another that evaluates that matrix element at a fixed value of $s$.

The oracle for finding the column index for a specified element in row $x$ of $\mathbf{H}$ can be constructed as follows.  The oracle $O_1$ has the
property that
\begin{equation}
O_1(q)\ket{x} \ket{y(q)},
\end{equation}
where $y(q)$ is the column index of the $q^{\rm th}$ element in row $x$.  Then for any $j$ we can construct the corresponding oracle by exploiting the
block diagonal structure of $\mathbf{H}$ via
\begin{equation}
\mathbf{O}_1(q) \ket{j}\ket{x} \ket{0} = \ket{j}\ket{x} \ket{y(q)+2^n(j-1)}=\ket{j}\ket{x} \ket{j}\ket{y(q)}.
\end{equation}
The oracle $\mathbf{O}_1(q)$ can therefore be enacted using one query to $O_1$ and a polynomial size arithmetic circuit.

The second oracle $\mathbf{O}_2(q)$ gives for a specific value of $s$ that is specified, the value of $H(s)$.  Specifically, after taking into account
the block diagonal structure of $\mathbf{H}$, we need the oracle to be of the form
\begin{equation}
\mathbf{O}_2 \ket{j}\ket{x}\ket{y}\ket{s} \ket{0} = \ket{j}\ket{x}\ket{y}\ket{[{H}(f_j(s))]_{x,y}}.
\end{equation}
This oracle can be implemented using one query to $O_f$ and one query to $O_2$.

In~\cite{WBHS11}, it is assumed that the time is provided to the oracles via classical control.  Here, we assume that the time is provided via a
quantum register so we must add the cost of preparing this register to the cost, $C$, of simulating a one--sparse matrix.  Lemma 9 of~\cite{WBHS11}
gives us that, the query complexity (costed at 1/per bit of output yielded by $\mathbf{O}_1$ and $\mathbf{O}_2$) is
\begin{equation}
C\le 4n(z_n +2) + 3n_H.\label{eq:Ceq}
\end{equation}
For each one--sparse Hamiltonian that appears in the Trotter--Suzuki decomposition, the time register must be initialized once~\cite{WBHS11}.  This
causes an additional source of error and if we are to fit it within our error budget, we must reduce the error in other parts of the simulation
algorithm.  There are three sources of error in the simulation algorithm: Trotter--Suzuki error, error due to finite $n_H$ and error due to finite
$n_T$ (we have neglected errors in synthesizing single qubit operations etc).  Each of these three sources of error is chosen to be at
most~$\epsilon/3$ in~\cite{WBHS11}.  Therefore, if we reduce the error tolerance in the Trotter--Suzuki approximation to $\epsilon/6$ and allow an
error tolerance of $\epsilon/6$ for approximating $f_j(s)$ then the overall error will remain at most $\epsilon$.  Thus the overall complexity becomes
\begin{align}
\label{eq:cor6}
\Noracle&\le 12 \Costvar Md^2 5^{k-1}\left[(L+1)+ 24kd^2\Lambdavar \max_j T_j\left(\frac{5}{3} \right)^{k}\left(\frac{6d^2\Lambdavar\max_j
T_j}{(\epsilon/6)}\right)^{1/2k}\right].
\end{align}

The error in $e^{-i \mathbf{H}(s) T}$ is at most $\|\Delta \mathbf{H}(s)\| T$~\cite{NC00}, where $\Delta \mathbf{H}(s)$ is the error in implementing
the Hamiltonian.  By Taylor's theorem this is at most $\Gamma \max_j|\Delta f_j(s)| \max_j T_j$, where $\Delta f_j(s)$ is the error incurred by
approximating $f_j(s)$ to a finite number of digits. Let us define the number of digits used to express $f_j$ as $n_{f_j}$.  Then the error in $f_j$
is $\Delta f_j \le 2^{-n_{f_j}}$.  Hence it suffices to choose
\begin{equation}
2^{-n_{f_j}} \Gamma \max_j T_j = \epsilon/6.\label{eq:erroreq}
\end{equation}
Thus since we have to both compute the value of $f_j(s)$ to $n_{f_j}$ bits of precision using queries to $O_2$ and then uncompute it,
equations~\eq{erroreq} and~\eq{Ceq} give
\begin{equation}
C\le 4n(z_n +2) + 3n_H+2\left\lceil\log_2 \left(\frac{6\Gamma \max_j T_j}{\epsilon}\right) \right\rceil,
\end{equation}
as claimed.
\end{proof}
We therefore see from~\thm{efficient} that this model of adiabatic computation can be efficiently simulated using the posited oracles under reasonable
smoothness assumptions.  This naturally leads to the following corollary:

\begin{corollary}
Let $f_j(s):j=1,\ldots, p$ efficiently computable functions, $H(s)=\sum_{\mu=1}^M T_\mu H_\mu(s) T^\dagger_\mu$ where each $H_\mu(f_j(s))$ is a
\sparse{d} row--computable matrix for all $s$ and $p\in O({\rm poly}(n))$.  If the conditions of~\thm{efficient} are satisfied for the adiabatic paths
$\{f_1,\ldots,f_p\}$ then controlled adiabatic evolution under $\{H(f_1(s)), \ldots, H(f_p(s))\}$ is polynomially equivalent to both the circuit model
and in turn adiabatic quantum computation.
\end{corollary}
\begin{proof}
We know that a circuit simulation of the controlled adiabatic evolution is efficient under the assumptions of~\thm{efficient} given access to the
oracles $O_1$, $O_2$ and $O_f$.  If $H(s)$ is row computable, then it implies that there exist efficient algorithms to find the locations and values
of each non--zero matrix element of $H(s)$.  Thus $O_1$  and $O_2$ can be implemented efficiently by the definition of row computability.

$O_f$ can be efficiently computed for each $j$ by assumption and hence for any fixed $j$ and $s$ the state $\ket{f_j(s)}$ can be prepared efficiently.
 Furthermore, because $p\in O({\rm poly}(n))$, it follows that the state $\sum_{j=1}^p a_j \ket{j} \ket{f_j(s)}$ can be efficiently prepared.  Thus
$O_f$ can be efficiently simulated in the circuit model as well.  This implies that quantum computers can efficiently simulate this class of
coherently controlled adiabatic evolutions.

Local Hamiltonians are a subset of \sparse{d} Hamiltonians.  Therefore the class of adiabatic evolutions considered includes a set of Hamiltonians
that generate a family of evolutions that are polynomially equivalent to the circuit model~\cite{AVD+08}.  Thus if we ignore the control register then
the controlled adiabatic evolution can be reduced to a universal adiabatic quantum computer.  Thus our model of computation is polynomially equivalent
to both the circuit model and adiabatic quantum computation. 
\end{proof}

We now see that controlled adiabatic quantum computation using piecewise smooth, sparse, bounded, row--computable Hamiltonians is not an exponentially
more powerful model of computation than traditional adiabatic computation.  Apart from showing that this is a reasonable model of quantum computation,
it also shows that the maximum evolution time used is a reasonable metric for the cost of the evolution (once made dimensionless by multiplying by a
characteristic energy of the system).  For most of the adiabatic paths considered, the contribution of the derivatives of the Hamiltonian to $\Lambda$
is negligible and thus in practice it suffices to ignore their contributions.  Also, because this algorithm scales near--linearly with the evolution
time, this analysis clearly shows that our model can only potentially provide sub--polynomial speedups over circuit based quantum computation for
fixed $d$ and $M$.
\section{conclusion}
The big question that our work addresses is: does coherent control over an adiabatic state preparation protocol yield any concrete benefit?  We find
that it does.  Specifically, it allows us to combine the best features of local adiabatic evolution and boundary cancellation methods, which are two
optimization strategies that are traditionally at odds with each other.  These combined strategies provide better error scaling than any known
adiabatic optimization technique for evolution times that are close to the transition between the Landau--Zener regime and the adiabatic regime.  We
finally provide an explicit quantum simulation algorithm for simulating our protocol, and in turn traditional adiabatic algorithms, that explicitly
gives the cost of simulating the controlled adiabatic evolution using a quantum computer and find that this cost scales near--linearly in the
evolution time.

These results only begin to scratch the surface of what is possible within this paradigm.  Our approach explicitly uses linear combinations of unitary
operations that are nearly unitary.  This raises an interesting question of whether truly non--unitary processes will be of use in optimizing
adiabatic passage.  Progress towards this goal has already been reported in~\cite{xu+14}.  Also, techniques similar to ours may be of value in phase
randomization protocols similar to~\cite{BKS09}.  Ultimately, these ideas may even lead to more natural ways of performing error correction or
suppression in a coherently controlled adiabatic quantum computer.  These are just a few examples of the many avenues of research that are opened by
this work.

\acknowledgements
This research was carried out while MK was 
visiting the Institute for Quantum Computing at the University of Waterloo and MK  is grateful for their kind hospitality.  We  thank M. Mosca and I.
Hen for valuable comments and feedback.  NW acknowledges
funding from USARO-DTO, NSERC and CIFAR. MK acknowledges support
from APVV QUTE. 

\appendix

\section{Suppressing both transitions in a three level system}\label{app:threelevel}
The approaches used to cancel the first order transitions for a two level system
can be generalized for larger systems as well. In a case of a three
level system, we must ensure that transitions to both first and second excited
state are $O(1/T)$. This occurs if the following conditions are met
\begin{align}
0=&\cos^2{(\theta)}  
\left[\frac{\braket{\dot{e}_1(1)}{g(1)}}{-i\gamma^A_{g,e_1}\left(1\right)T_A} -
e^{i\int_0^1\gamma^A_{g,e_1}(\xi)d\xi T_A}
\frac{\braket{\dot{e_1}(0)}{g(0)}}{-i\gamma^A_{g,e_1}\left(0\right)T_A}\right],
\nonumber\\
&+\sin^2{(\theta)}
\left[\frac{\braket{\dot{e}'_1(1)}{g(1)}}{-i\gamma^B_{g,e_1}\left(1\right)T_B} -
e^{i\int_0^1\gamma^B_{g,e_1}(\xi)d\xi T_B}
\frac{\braket{\dot{e_1}(0)}{g(0)}}{-i\gamma^B_{g,e_1}\left(0\right)T_B}\right]
\label{first_excitation}\\
0=&\cos^2{(\theta)}
\left[\frac{\braket{\dot{e_2}(1)}{g(1)}}{-i\gamma^A_{g,e_2}\left(1\right)T_A}
-e^{i\int_0^1\gamma^A_{g,e_2}(\xi)d\xi
T_A}
\frac{\braket{\dot{e_2}(0)}{g(0)}}{-i\gamma^B_{g,e_1}\left(0\right)T_A}\right]
\nonumber\\
&+\sin^2{(\theta)}\left[\frac{\braket{\dot{e}'_2(1)}{g(1)}}{-i\gamma^B_{g,e_2}
\left(1\right)T_B} -e^{i\int_0^1\gamma^B_{g,e_2}(\xi)d\xi
T_B}
\frac{\braket{\dot{e}'_2(0)}{g(0)}}{-i\gamma^B_{g,e_2}\left(0\right)T_B}\right],
\label{second_excitation}.
\end{align}
where the the first evolution corresponds to a Hamiltonian $H(f^A(s))$ and the
second one to $H(f^B(s))$ with its states denoted by primes,
parameterizing Hamiltonian with a single function as in the last paragraph.
Moreover, we assume that $H(f^A(s))$ and $H(f^B(s))$ are equal at the
beginning and the end of evolution (but their derivatives with respect to s
differ). 

It is straightforward to cancel transitions at certain (discrete) times using our
knowledge from the $2$ level case when we realize
\begin{align}
\braket{\dot{e}_1(0)}{g(0)}&=\mvalue{e_1(0)}{\dot{f}(s)|_{s=0}H_1}{g(0)}\\
\braket{\dot{e}_2(0)}{g(0)}&=\mvalue{e_2(0)}{\dot{f}(s)|_{s=0}H_1}{g(0)}\\
\braket{\dot{e}_1(1)}{g(1)}&=\mvalue{e_1(1)}{-\dot{f}(s)|_{s=1}H_0}{g(1)}\\
\braket{\dot{e}_2(1)}{g(1)}&=\mvalue{e_2(1)}{-\dot{f}(s)|_{s=1}H_0}{g(1)}.
\end{align}
Hence, by choosing $f^A(s)$ and $f^B(s)$ as in \sec{partial} and using \eq{2level_weight}, we get rid of terms
containing derivatives at the end for both levels. This approach trivially generalizes to higher--dimensional systems.

Now we need to fix $T_A$ and $T_B$ in order to remove the end derivatives parts
by requiring that evolutions gain opposite phases. We can rewrite
already simplified \eqref{first_excitation}, \eqref{second_excitation} as 
\begin{equation}
 \begin{pmatrix}
  \int_0^1\gamma^A_{g,e_1}(\xi)d\xi & \int_0^1\gamma^B_{g,e_1}(\xi)d\xi \\
  \int_0^1\gamma^A_{g,e_2}(\xi)d\xi & \int_0^1\gamma^B_{g,e_2}(\xi)d\xi
 \end{pmatrix}
  \begin{pmatrix}
  T_A \\
  T_B \\
 \end{pmatrix}
 =
  \begin{pmatrix}
  (2n+1)\pi \\
  (2m+1)\pi
 \end{pmatrix}.
\end{equation}
This system of equation has a solution, unless the determinant of the matrix
equals zero. Note, that with this approach we get only a discrete
set of times $T_A$ and $T_B$ for which the error vanishes in contrast to many of our prior methods.

Error suppression can also be achieved for arbitrary time if we use more than 2
evolutions. A 2-level inspired solution uses 4 unitaries $U_A$, $U_B$,
$U_C$ and $U_D$ where $U_A$ and $U_C$ are given by Hamiltonian $H(f^A(s))$ and
$U_B$ and $U_D$ by $H(f^B(s))$. We pick the functions $f^A,f^B$ based on \sec{partial}. The
goal is then to find times $T_A$-$T_D$ and weights $a$-$d$ such that
following equations are satisfied
\begin{align}
0&=\frac{\braket{\dot{e}_1(1)}{g(1)}}{\gamma^A_{g,e_1}\left(1\right)}\left[
\frac{a}{T_A}-\frac{b}{T_B}+\frac{c}{T_C}- \frac{d}{T_D}\right]\nonumber\\
&+ \frac{\braket{\dot{e_1}(0)}{g(0)}}{-i\gamma^A_{g,e_1}\left(0\right)}\left[
\frac{a}{T_A} e^{i\int_0^1\gamma^A_{g,e_1}(\xi)d\xi T_A} +  \frac{b}{T_B}
e^{i\int_0^1\gamma^B_{g,e_1}(\xi)d\xi T_B} + \frac{c}{T_C}
e^{i\int_0^1\gamma^A_{g,e_1}(\xi)d\xi T_C} + \frac{d}{T_D}
e^{i\int_0^1\gamma^B_{g,e_1}(\xi)d\xi T_D}\right] \\
0&=\frac{\braket{\dot{e}_2(1)}{g(1)}}{\gamma^A_{g,e_2}\left(1\right)}\left[
\frac{a}{T_A}-\frac{b}{T_B}+\frac{c}{T_C}- \frac{d}{T_D}\right]\nonumber\\
&+ \frac{\braket{\dot{e_2}(0)}{g(0)}}{-i\gamma^A_{g,e_2}\left(0\right)}\left[
\frac{a}{T_A} e^{i\int_0^1\gamma^A_{g,e_2}(\xi)d\xi T_A} +  \frac{b}{T_B}
e^{i\int_0^1\gamma^B_{g,e_2}(\xi)d\xi T_B} + \frac{c}{T_C}
e^{i\int_0^1\gamma^A_{g,e_2}(\xi)d\xi T_C} + \frac{d}{T_D}
e^{i\int_0^1\gamma^B_{g,e_2}(\xi)d\xi
T_D}\right] 
\end{align}

First, the normalization condtion
\begin{equation}
a+b+c+d=1 \label{norm}
\end{equation}
must hold. 
Second, we choose $b$ and $T_B$ such that they cancel the error on first level
from evolution by $U_A$. This is exactly the same problem we solved for 
a 2 level system, hence the proper $b$ and $T_B$ are 
\begin{align}
T_B&=\frac{\int_0^1 \gamma^A_{g,e_1}\left(\xi\right)d \xi T_A
+\left(2n+1\right)\pi}{\int_0^1 \gamma^B_{g,e_1}\left(\xi\right) d \xi}\\
b&=\frac{a T_B}{T_A}.
\end{align}

The same can be done for $U_C$ and $U_D$:
\begin{align}
T_D&=\frac{\int_0^1 \gamma^A_{g,e_1}\left(\xi\right)d \xi T_C
+\left(2n+1\right)\pi}{\int_0^1 \gamma^B_{g,e_1}\left(\xi\right) d \xi}\\
d&=\frac{c T_D}{T_C}.
\end{align}
This suppresses the first transition out of $\ket{g}$ (typically the transition to the first excited state) to $O(1/T^2)$. In addition, the errrors
from the derivatives
at the end on the
second excited state cancel as well.

Finally, we are left with
\begin{align}
0=&\frac{a}{T_A}\left[ e^{i\int_0^1\gamma^A_{g,e_2}(\xi)d\xi T_A} +
e^{\frac{\int_0^1\gamma^B_{g,e_2}(\xi)d\xi}{\int_0^1\gamma^B_{g,e_1}(\xi)d\xi}
\left(
\int_0^1 \gamma^A_{g,e_1}\left(\xi\right)d\xi T_A +\left(2n+1\right)\pi
\right)}\right]\nonumber\\
&+\frac{c}{T_C}\left[ e^{i\int_0^1\gamma^A_{g,e_2}(\xi)d\xi T_C} +
e^{\frac{\int_0^1\gamma^B_{g,e_2}(\xi)d\xi}{\int_0^1\gamma^B_{g,e_1}(\xi)d\xi}
\left(
\int_0^1 \gamma^A_{g,e_1}\left(\xi\right)d\xi T_C +\left(2m+1\right)\pi
\right)}\right]. \label{3levels}
\end{align}
Therefore we can set the value of $T_C$ and ratio of $a$ and $c$ and we are still
free to choose arbitrary $T_A$. After some algebra we can rewrite \eqref{3levels}
\begin{align}
&\frac{a}{T_A} e^{ \frac{i}{2}\left(\int_0^1\gamma^A_{g,e_2}(\xi)d\xi T_A +
\frac{\int_0^1\gamma^B_{g,e_2}(\xi)d\xi}{\int_0^1\gamma^B_{g,e_1}(\xi)d\xi}
\left(\int_0^1 \gamma^A_{g,e_1}\left(\xi\right)d\xi T_A +\left(2n+1\right)\pi\right)
\right)} \nonumber\\
&\qquad\times\cos{\left[ 
\frac{1}{2}\left(\int_0^1\gamma^A_{g,e_2}(\xi)d\xi T_A +
\frac{\int_0^1\gamma^B_{g,e_2}(\xi)d\xi}{\int_0^1\gamma^B_{g,e_1}(\xi)d\xi}
\left(\int_0^1 \gamma^A_{g,e_1}\left(\xi\right)d\xi T_A +\left(2n+1\right)\pi
\right)\right)
\right]}\nonumber\\
&=-\frac{c}{T_C} e^{ \frac{i}{2}\left(\int_0^1\gamma^A_{g,e_2}(\xi)d\xi T_C +
\frac{\int_0^1\gamma^B_{g,e_2}(\xi)d\xi}{\int_0^1\gamma^B_{g,e_1}(\xi)d\xi}
\left(\int_0^1 \gamma^A_{g,e_1}\left(\xi\right)d\xi T_C +\left(2m+1\right)\pi\right)
\right)} \nonumber\\
&\qquad\times\cos{\left[
\frac{1}{2}\left(\int_0^1\gamma^A_{g,e_2}(\xi)d\xi T_C +
\frac{\int_0^1\gamma^B_{g,e_2}(\xi)d\xi}{\int_0^1\gamma^B_{g,e_1}(\xi)d\xi}
\left(\int_0^1 \gamma^A_{g,e_1}\left(\xi\right)d\xi T_C +\left(2m+1\right)\pi
\right)\right)
\right]}.
\end{align}
We can ensure that both sides of the equation pick the same phases (up to $2k\pi$) by setting $T_C$ and then
we only need a value of $k$ for which the cosines would have the same sign. That is possible, unless
\begin{equation}
\int_0^1\gamma^A_{g,e_2}(\xi)d\xi +
\frac{\int_0^1\gamma^B_{g,e_2}(\xi)d\xi}{\int_0^1\gamma^B_{g,e_1}(\xi)d\xi}
\int_0^1 \gamma^A_{g,e_1}\left(\xi\right)d\xi\equiv 0 \mod 2\pi.
\end{equation}  
 This procedure can also be used to find paths that cancel multiple transitions in higher--dimensional systems but a closed form may not necessarily
exist, unlike the present case.

\end{document}